\documentclass[10pt,lettersize,journal,twocolumn,twoside]{IEEEtran}

\usepackage{textcomp}
\usepackage{stfloats}

\usepackage{graphicx,amsmath,amsfonts,amsthm,color,comment}
\usepackage{enumitem}
\usepackage{enumerate}
\usepackage{algorithmic}
\usepackage{pgfplots}
\usepackage{pgf}
\usepackage{tikz}
\usepackage{amssymb}
\usepackage{csquotes}
\usepackage{bbm}
\usepackage{multirow}
\usepackage{bm}
\usepackage[nolist]{acronym}
\usepackage{amsthm}
\theoremstyle{definition}
\newtheorem{theorem}{Theorem}

\DeclareMathOperator*{\argmin}{arg\,min}
\DeclareMathOperator*{\argmax}{arg\,max}

\usepackage[ruled,linesnumbered,vlined]{algorithm2e}
\SetNlSty{bfseries}{\color{black}}{}
\SetArgSty{textnormal}
\pgfplotsset{compat=1.15}
\SetKw{Continue}{continue}
\SetKw{Break}{break}

\def\ben{\begin{eqnarray*}}
\def\een{\end{eqnarray*}}

\def  \H5G{H_\text{CA-Polar}}

\begin{document}


\title{On the Role of Quantization of Soft Information in GRAND}



\author{\IEEEauthorblockN{Peihong Yuan, \IEEEmembership{Member, IEEE}, Ken R. Duffy, \IEEEmembership{Member, IEEE}, Evan P. Gabhart, \IEEEmembership{Student Member, IEEE}, \\and Muriel M\'edard, \IEEEmembership{Fellow, IEEE}}, 
	\thanks{This work was supported by the Defense Advanced Research Projects Agency under Grant HR00112120008.}
	\thanks{Peihong Yuan, Evan P. Gabhart and Muriel M\'edard are with the Research Laboratory of Electronics, Massachusetts Institute of Technology, Cambridge MA 02139 USA (e-mail: phyuan@mit.edu; egabby97@mit.edu; medard@mit.edu).}
	\thanks{Ken R. Duffy is with Hamilton Institute, Maynooth University, W23 Maynooth, Ireland (e-mail: ken.duffy@mu.ie).}
}

\begin{acronym}
    \acro{BCH}{Bose-Chaudhuri-Hocquenghem}
    \acro{GRAND}{guessing random additive noise decoding}
    \acro{SGRAND}{soft GRAND}
    \acro{DSGRAND}{discretized soft GRAND}
    \acro{SRGRAND}{symbol reliability GRAND}
    \acro{ORBGRAND}{ordered reliability bits GRAND}
    \acro{5G}{the $5$-th generation wireless system}
	\acro{APP}{a-posteriori probability}
	\acro{ARQ}{automated repeat request}
	\acro{ASK}{amplitude-shift keying}
	\acro{AWGN}{additive white Gaussian noise}
	\acro{B-DMC}{binary-input discrete memoryless channel}
	\acro{BEC}{binary erasure channel}
	\acro{BER}{bit error rate}
	\acro{biAWGN}{binary-input additive white Gaussian noise}
	\acro{BLER}{block error rate}
	\acro{uBLER}{undetected block error rate}
	\acro{bpcu}{bits per channel use}
	\acro{BPSK}{binary phase-shift keying}
	\acro{BSC}{binary symmetric channel}
	\acro{BSS}{binary symmetric source}
	\acro{CDF}{cumulative distribution function}
	\acro{CRC}{cyclic redundancy check}
	\acro{DE}{density evolution}
	\acro{DMC}{discrete memoryless channel}
	\acro{DMS}{discrete memoryless source}
	\acro{BMS}{binary input memoryless symmetric}
	\acro{eMBB}{enhanced mobile broadband}
	\acro{FER}{frame error rate}
	\acro{uFER}{undetected frame error rate}
	\acro{FHT}{fast Hadamard transform}
	\acro{GF}{Galois field}
	\acro{HARQ}{hybrid automated repeat request}
	\acro{i.i.d.}{independent and identically distributed}
	\acro{LDPC}{low-density parity-check}
	\acro{LHS}{left hand side}
	\acro{LLR}{log-likelihood ratio}
	\acro{MAP}{maximum-a-posteriori}
	\acro{MC}{Monte Carlo}
	\acro{ML}{maximum-likelihood}
	\acro{PDF}{probability density function}
	\acro{PMF}{probability mass function}
	\acro{QAM}{quadrature amplitude modulation}
	\acro{QPSK}{quadrature phase-shift keying}
	\acro{RCU}{random-coding union}
	\acro{RHS}{right hand side}
	\acro{RM}{Reed-Muller}
	\acro{RV}{random variable}
	\acro{RS}{Reed–Solomon}
	\acro{SCL}{successive cancellation list}
	\acro{SE}{spectral efficiency}
	\acro{SNR}{signal-to-noise ratio}
	\acro{UB}{union bound}
	\acro{BP}{belief propagation}
	\acro{NR}{new radio}
	\acro{CA-SCL}{CRC-assisted successive cancellation list}
	\acro{DP}{dynamic programming}
	\acro{URLLC}{ultra-reliable low-latency communication}
\end{acronym}

\maketitle

\begin{abstract}
In this work, we investigate guessing random additive noise decoding (GRAND) with quantized soft input. First, we analyze the achievable rate of ordered reliability bits GRAND (ORBGRAND), which uses the rank order of the reliability as quantized soft information. We show that multi-line ORBGRAND can approach capacity for any signal-to-noise ratio (SNR). 
We then introduce discretized soft GRAND (DSGRAND), which uses information from a conventional quantizer. Simulation results show that DSGRAND well approximates maximum-likelihood (ML) decoding with a number of quantization bits that is in line with current soft decoding implementations. For a $(128,106)$ CRC-concatenated polar code, the basic ORBGRAND is able to match or outperform CRC-aided successive cancellation list (CA-SCL) decoding with codeword list size of $64$ and $3$ bits of quantized soft information, while DSGRAND outperforms CA-SCL decoding with a list size of $128$ codewords. Both ORBGRAND and DSGRAND exhibit approximately an order of magnitude less average complexity and two orders of magnitude smaller memory requirements than CA-SCL.
\end{abstract}

\markboth
	{submitted to IEEE Transactions on Signal Processing}
	{}

\section{Introduction}

As \ac{ML} error correcting decoding of linear codes is NP-complete~\cite{berlekamp1978inherent},
the engineering paradigm has been to co-design restricted classes of linear codebooks with code-specific decoding methods that exploit the code structure to enable computationally efficient approximate-\ac{ML} decoding. For example, \ac{BCH} codes with hard detection Berlekamp-Massey decoding~\cite{berlekamp1968algebraic,massey1969shift}, \ac{LDPC} codes~\cite{gallager1963low} and \ac{BP}~\cite{chen2002near}, polar codes with \ac{CRC} codes, which have been selected for all control channel communications in \ac{5G} \ac{NR}, and \ac{CA-SCL} decoding~\cite{niu2012crc,tal2015list,balatsoukas2015llr,leonardon2019fast}.

Modern applications, including augmented and virtual reality, vehicle-to-vehicle communications, the Internet of Things, and machine-type communications, have driven demand for reliable, low-atency communications~\cite{durisi2016toward,she2017radio,chen2018ultra,parvez2018survey,medard20205}. These technologies benefit from short (for latency) and high-rate (for efficiency) codes, reviving the possibility of creating high-accuracy universal decoders that are suitable for hardware implementation. Accurate, practically realizable universal decoders offer the possibility of reduced hardware footprint, the provision of hard or soft detection decoding for codes that currently only have one class of decoder, and future-proofing devices  against the introduction of new codes.

\Ac{GRAND} is a recently established universal decoder that was originally introduced for hard decision demodulation systems~\cite{Duffy18,duffy19GRAND}. \Ac{GRAND} algorithms operate by sequentially removing putative noise effects from the demodulated received sequence and querying if what remains is in the codebook. The first instance where a codebook member is found is the decoding. If \ac{GRAND} queries noise effects from most likely to least likely based on available hard or soft information, it identifies a \ac{ML} decoding. Amongst other results, in the hard detection setting, mathematical analysis of \ac{GRAND} determines error exponents for a broad class of additive noise models~\cite{duffy19GRAND}.

For an $(n,k)$ code, where $k$ information bits are transformed into $n$ coded bits for
communication, all \ac{GRAND} algorithms identify an erroneous decoding after an approximately geometrically distributed number of codebook queries with mean $2^{n-k}$~\cite[Theorem 2]{duffy19GRAND} and correctly decode if they identify a codeword beforehand. As a result, an upper bound on the complexity of all \ac{GRAND} algorithms is determined by the number of redundant bits rather than the code length or rate directly, making them suitable for decoding any moderate redundancy code of any length. The performance difference between \ac{GRAND} variants stems from their utilisation of statistical noise models or soft information to improve the targeting of their queries.

The evident parallelizability of hard detection \ac{GRAND}'s codebook queries have already resulted in the proposal~\cite{abbas2020} and realization~\cite{Riaz21,Riaz22} of efficient circuit implementations for \acp{BSC}. An algorithm has also been introduced for channels subject to bursty noise whose statistical characteristics are known to the receiver~\cite{an20,an2022keep}, which has also resulted in proposed circuit implementations~\cite{abbas2021high-MO,zhan2021noise}. 

A natural question is how to make use of soft detection information, when it is available, in order to improve \ac{GRAND}'s query order and several proposals have been made. \Ac{SGRAND}~\cite{solomon20SGRAND} uses real-valued soft information to build a bespoke noise-effect query order for each received signal and provides a benchmark for optimal decoding accuracy performance. It is possible to create a semi-parallelizable implementation in software using dynamic max-heap structures, but the requirement for substantial dynamic memory does not lend itself to efficient implementation in hardware.

Quantization of the soft information provides a trade-off between performance and implementation complexity. \Ac{SRGRAND}~\cite{Duffy19a,Duffy22} avails of the most limited quantized soft information where one additional bit tags each symbol as being reliably or unreliably received. \Ac{SRGRAND} is mathematically analysable, implementable in hardware, and provides a $0.5-0.75$~dB gain over hard detection \ac{GRAND}~\cite{Duffy22}. 

\Ac{ORBGRAND}~\cite{duffy2021ordered} aims to bridge the gap between \ac{SRGRAND} and \ac{SGRAND} by obtaining the decoding accuracy close to the latter in an algorithm that is suitable for implementation in circuits. For a block code of length $n$, it uses $\lceil\log_2(n)\rceil$ bits of codebook-independent quantized soft detection information per received bit, the rank order of each received bit's reliability, to determine an accurate decoding. It retains the hard detection algorithm's suitability for a highly parallelized implementation in hardware and high throughput VLSI designs have been proposed in~\cite{abbas2021high,condo2021high,condo2022fixed,abbas22}. Moreover, theoretical results indicate that \ac{ORBGRAND} is almost capacity achieving in certain \ac{SNR} regimes~\cite{liu2022orbgrand}. 
It has been observed that \ac{ORBGRAND} provides near-ML decoding for \acp{BLER} greater than $10^{-4}$, but is less precise at higher \ac{SNR}. The latter effect can be explained by the fact that \ac{ORBGRAND}'s noise effect query order diverges from the optimal rank order in terms of decreasing likelihood. To rectify this, a list decoding approach to improve \ac{ORBGRAND}'s performance at higher \ac{SNR} has been suggested~\cite{abbas2021list}. A more accurate statistical model based on multi-line fitting of rank orders from which an enhanced query order can be determined has also been introduced~\cite{duffy22ORBGRAND}.

In this work, we start with the achievable rate of \ac{ORBGRAND}. We construct a mismatched demapping function to approximate the reliability sorting operation, and then calculate the achievable rate of the corresponding mismatched decoding. We show that multi-line \ac{ORBGRAND} can be capacity achieving at all \acp{SNR} as long as a sufficiently rich multi-line fitting is employed, expanding upon the results of~\cite{liu2022orbgrand} that show that the basic version of \ac{ORBGRAND}~\cite{duffy2021ordered} is near-capacity achieving in certain SNR regimes. The number of lines in practice we show to be quite moderate to achieve near-optimal performance.

We also propose \ac{DSGRAND}, an alternative efficient algorithm based on \ac{DP} for soft detection error correction decoding with \ac{GRAND} that is an approximation to \ac{SGRAND}. In contrast to \ac{ORBGRAND}, \ac{DSGRAND} uses conventional quantizers and does not require received bits to be rank-ordered by their reliability, which can be energy expensive when implemented in circuits. \Ac{DSGRAND} can make use of any level of quantized soft detection information, with increasingly accurate performance as the number of soft information bits per received bit increases. Quantizer optimization for \Ac{DSGRAND} yields to theoretical analysis and heuristic designs, both of which we discuss. Empirical results show \Ac{DSGRAND}  outperforms the original basic \ac{ORBGRAND} at high \ac{SNR} and performs within $0.25~\text{dB}$ and $0.1~\text{dB}$ close to \ac{ML} decoding with $2$ and $3$ bits quantizers, respectively. We also provide the comparison of performance and complexity between \ac{GRAND} variants and \ac{CA-SCL} decoding for a $(128,106)$ 5G \ac{CRC}-concatenated polar code. With a $3$ bits quantizer, proposed \ac{DSGRAND} performs approximately $0.2~\text{dB}$ better than a \ac{CA-SCL} decoder ($L=128$) with significantly lower average complexity and memory requirements.

This paper is organized as follows. Section~\ref{sec:prelim} gives background on the problem. The achievable rate of \ac{ORBGRAND} is discussed in Section~\ref{sec:orbgrand}. Section~\ref{sec:qgrand} presents the proposed \ac{DSGRAND} algorithm with pseudo code and the quantizer optimization. Simulation results are shown in Section~\ref{sec:numerical}. Section~\ref{sec:conclusions} concludes the paper.

\section{Preliminaries}\label{sec:prelim}
In this paper, vectors are written as $x^n=\left(x_1,x_2,\dots,x_n\right)$. The all-zeros vector of dimension $n$ is $0^n$. The $i$-th entry of $x^n$ is $x_i$. A \ac{RV} is written with an uppercase letter such as $X$. A realization of $X$ is written with the corresponding lowercase letter $x$. A vector of \ac{RV}s is written as $X^n=\left(X_1,X_2,\dots,X_n\right)$. Both the \ac{PDF} of a continuous \ac{RV} and the \ac{PMF} of a discrete \ac{RV} evaluated at $x$ are written as $f_X(x)$. $F_X(x)$ denotes the \ac{CDF} of \ac{RV} $X$. The bit-wise \enquote{exclusive or} operation is written as $\oplus$. For functions and operations defined with scalar inputs, we use them with vector inputs as their element-wise version, i.e.,
\begin{align*}
f(x^n) &= \left(f(x_1),f(x_2),\dots,f(x_n)\right),\\
x^n \oplus y^n &= \left(x_1 \oplus y_1,x_2 \oplus y_2,\dots,x_n \oplus y_n\right).
\end{align*}

\subsection{Channel Model and ML Decoding}
Consider a communication system using a binary $(n,k)$ linear block code. The codeword $c^n$ is \ac{BPSK} modulated with unit power via
\begin{align}\label{eq:bpsk}
    f_\text{BPSK}(c_i) = 1-2c_i,~i=1,\dots,n
\end{align}
and transmitted over a memoryless \ac{AWGN} channel, i.e.,
\begin{align}\label{eq:biawgn}
    Y = f_\text{BPSK}(C) + \sigma N
\end{align}
where $N$ is an independent zero mean Gaussian noise with variance one. The \ac{SNR} is given by
\begin{align}
    \frac{E_s}{N_0}=\frac{1}{\sigma^2},\quad\frac{E_s}{N_0}=\frac{E_b}{N_0}\cdot \frac{2k}{n}.
\end{align}

At the receiver, we have the symbol-wise \ac{LLR} of $c_i$ based on the observation $y_i,~i=1,\dots,n$
\begin{align}
    \tau(y_i) \triangleq\log\frac{f_{Y|C}(y_i|0)}{f_{Y|C}(y_i|1)}=\frac{2y_i}{\sigma^2},~i=1,\dots,n.
    \label{eq:LLR}
\end{align}
We also have the hard decision $\tilde{c}_i$ and the symbol-wise reliabilities $\ell_i$ of the coded bit $c_i$,
\begin{align}
    \tilde{c}_i \triangleq \frac{1-\text{sign}\left(y_i\right)}{2},\quad \ell_i\triangleq\left|\tau(y_i)\right|=\frac{2\left|y_i\right|}{\sigma^2}.
\end{align}
A block-wise \ac{ML} decoder puts out the codeword
\begin{align}
    \hat{c}^n &= \argmax_{c^n\in\mathcal{C}} f_{Y^n|C^n}\left(y^n|c^n\right)\\
    &= \argmax_{c^n\in\mathcal{C}} \prod_{i=1}^n f_{Y|C}\left(y_i|c_i\right)\\
    &=\argmax_{c^n\in\mathcal{C}} \sum_{i=1}^n f_\text{BPSK}(c_i) \cdot \tau\left(y_i\right)\\
    &=\argmax_{c^n\in\mathcal{C}} \sum_{i=1}^n (1-2c_i) \cdot \tau\left(y_i\right).\label{eq:MLD}
\end{align}

\subsection{Reliability quantization}\label{sec:quantizer}
A conventional $q$ bits ($Q=2^q$-level) reliability quantizer $\mu(\cdot):[0,\infty)\mapsto[0,\infty)$ maps a positive input value in range $[b_{i-1},b_i)$ to a output value $v_i$ for $i=1,\dots,Q$, where $b_0=0$ and $b_Q=+\infty$, i.e.,
\begin{equation}\label{eq:quantizer}
  \mu(x)=\left\{
    \begin{aligned}
      v_1, &\text{ if } x\in\left[0,b_1\right)\\
      v_2, &\text{ if } x\in\left[b_1,b_2\right)\\
      \vdots &\\
      v_Q, &\text{ if } x\in\left[b_{Q-1},+\infty\right).
    \end{aligned}\right.
\end{equation}
Note that $q$ bits reliability quantization equivalently quantizes the \ac{LLR} $\tau(y)$ and the channel output $y$ with $q+1$ bits, i.e., one extra bit is required for the sign (or equivalently the hard decision $\tilde{c}$). To simplify  explanations in future sections, we define the equivalent quantizer $\bar{\mu}(\cdot)$ for \acp{LLR},
\begin{equation}
  \bar{\mu}(\tau(y))\triangleq\left\{
    \begin{aligned}
      \mu(\tau(y)), &\text{ if } \tau(y)\geq 0\\
      -\mu(-\tau(y)), &\text{ otherwise }
    \end{aligned}\right.
\end{equation}
and $\tilde{\mu}(\cdot)$ for channel output $y$,
\begin{equation}
  \tilde{\mu}(y)\triangleq \bar{\mu}\left(\frac{2y}{\sigma^2}\right)\cdot \frac{\sigma^2}{2}.
\end{equation}

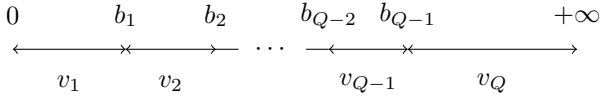
\begin{figure}[t]
	\centering
	\begin{tikzpicture}[scale=1.5]

\draw[<->] (0,0) to (1,0);
\draw[<->] (1,0) to (1.8,0);
\draw[] (1.8,0) to (2.0,0);
\node at (2.3,0) {$\cdots$};
\draw[] (2.6,0) to (2.8,0);
\draw[<->] (2.8,0) to (3.5,0);
\draw[<->] (3.5,0) to (5,0);

\node at (0,0.3) {$0$};
\node at (1,0.3) {$b_1$};
\node at (1.8,0.3) {$b_2$};
\node at (2.8,0.3) {$b_{Q-2}$};
\node at (3.5,0.3) {$b_{Q-1}$};
\node at (5,0.3) {$+\infty$};

\node at (0.5,-0.3) {$v_1$};
\node at (1.4,-0.3) {$v_2$};
\node at (3.15,-0.3) {$v_{Q-1}$};
\node at (4.25,-0.3) {$v_{Q}$};

\end{tikzpicture}
	\caption{Input-output relation of a quantizer for positive real numbers.}
	\label{fig:quantizer}
\end{figure}

\subsection{GRAND variants}\label{sec:GRANDs}
All \ac{GRAND}~\cite{duffy19GRAND} algorithms seek to identify the noise effect 
\begin{align}\label{eq:noiseeffect}
    z^n\triangleq c^n\oplus\tilde{c}^n,
\end{align}
that has impacted the transmission, from which the decoded codeword is deduced. \Ac{GRAND} creates binary error patterns $e^n$ rank-ordered by score $S(e^n)$ to find a valid codeword, i.e., $\hat{c}^n=\tilde{c}^n\oplus e^n \in\mathcal{C}$, i.e.,
\begin{align}\label{eq:GRAND}
    \hat{c}^n=\tilde{c}^n\oplus \hat{z}^n,\text{ where }\hat{z}^n &= \argmin_{e^n: \tilde{c}^n\oplus e^n\in\mathcal{C}} S(e^n)
\end{align}
where the score of an error pattern $e^n$ is defined by 
\begin{align}
    S(e^n)\triangleq\sum_{i=1}^n e_i\cdot w_i
\end{align}
where $w_i>0$ denotes the weight of the hard decision $\tilde{c}_i$, which could be also considered as the cost to flip $\tilde{c}_i$. Pseudo-code for GRAND is shown in Algorithm~\ref{alg:GRAND}. Note that scaling of the weights does not impact the decoding output of \acp{GRAND}, i.e.,
\begin{align}\label{eq:weightscaling}
    \argmin_{e^n: \tilde{c}^n\oplus e^n\in\mathcal{C}} \sum_{i=1}^n e_i\cdot\alpha\cdot w_i=\argmin_{e^n: \tilde{c}^n\oplus e^n\in\mathcal{C}} \sum_{i=1}^n e_i\cdot w_i,\text{~if~}\alpha>0.
\end{align}
\begin{algorithm}[t]
	\SetNoFillComment
	\DontPrintSemicolon
	\SetKwInOut{Input}{Input}\SetKwInOut{Output}{Output}
	\Input{hard decisions $\tilde{c}^n$, weight $w^n$}
	\Output{estimates $\hat{c}^n$, number of guesses $n_p$, decoding state $\phi$}
	\BlankLine
	$\phi\leftarrow\text{false}, n_p\leftarrow 1, e^n=0^n$\\
	\If{$\tilde{c}^n\in\mathcal{C}$}{
	    $\hat{c}^n\leftarrow \tilde{c}^n, \phi\leftarrow\text{true}$\\
	    \Return{}
	}
	\While{$\phi=\text{false}$}{
	    $e^n\leftarrow \text{next error pattern with least score $S(e^n)$}$\\
	    $n_p\leftarrow n_p+1$\\
	    \If{$\tilde{c}^n\oplus e^n \in\mathcal{C}$}{
    	   $\hat{c}^n\leftarrow \tilde{c}^n\oplus e^n,\phi\leftarrow\text{true}$\\
	       \Return{}
	    }
	}
	\Return{}
	\caption{$\text{GRAND}$}
	\label{alg:GRAND}
\end{algorithm}

In a hard decision BSC, \ac{GRAND}~\cite{Duffy18,duffy19GRAND} doesn't have any reliability information and thus uses weight $w_i=1$ for $i=1,\dots,n$, i.e., the score of error pattern $e^n$ is equal to its Hamming weight. If a bursty statistical channel characterization is available at the receiver, a Markov-informed order can be used to generate error patterns~\cite{an20,an2022keep}.

\Ac{SGRAND} uses the non-quantized reliability as the weight~\cite{solomon20SGRAND}, i.e.,
\begin{equation}
  w_i=\ell_i,\quad i=1,\dots,n.
\end{equation}

\begin{theorem}\label{th:sgrand}
SGRAND provides the \ac{ML} decision by
\begin{align}\label{eq:GRAND_ML}
    \hat{c}^n=\tilde{c}^n\oplus \hat{z}^n,\text{ where }\hat{z}^n &= \argmin_{e^n: \tilde{c}^n\oplus e^n\in\mathcal{C}} \sum_{i=1}^n e_i \cdot \ell_i.
\end{align}
\end{theorem}
\begin{proof} We check both possibilities $\tilde{c}_i=0,1$ for eq.~\eqref{eq:MLD},
\begin{itemize}
    \item If $\tau\left(y_i\right)\geq 0$ and $\tilde{c}_i=0$,
\begin{align}
(1-2c_i)\cdot \tau\left(y_i\right)&=(1-2(\tilde{c}_i\oplus e_i))\cdot \tau\left(y_i\right)\\
&=(1-2e_i)\cdot \ell_i.
\end{align}
\item If $\tau\left(y_i\right)< 0$ and $\tilde{c}_i=1$,
\begin{align}
(1-2c_i)\cdot \tau\left(y_i\right)&=(1-2(\tilde{c}_i\oplus e_i))\cdot \tau\left(y_i\right)\\
&=(1-2(1-e_i))\cdot -\ell_i\\
&=(1-2e_i)\cdot \ell_i.
\end{align}
\end{itemize}
We have an \ac{ML} decision by identify the noise effect $z^n$ via
\begin{align}
    \hat{z}^n &= \argmax_{e^n: \tilde{c}^n\oplus e^n\in\mathcal{C}} \sum_{i=1}^n (1-2(\tilde{c}_i\oplus e_i)) \cdot \tau\left(y_i\right)\\
    &= \argmax_{e^n: \tilde{c}^n\oplus e^n\in\mathcal{C}} \sum_{i=1}^n (1-2e_i)\cdot \ell_i\\
    &= \argmin_{e^n: \tilde{c}^n\oplus e^n\in\mathcal{C}} \sum_{i=1}^n e_i\cdot \ell_i.
\end{align}
By the definition of noise effect eq.~\eqref{eq:noiseeffect}, we have the \ac{ML} decision via $\hat{c}^n=\tilde{c}^n\oplus \hat{z}^n$. Note that this proof is equivalent to \cite[Theorem III.1]{solomon20SGRAND}.
\end{proof}

Theorem~\ref{th:sgrand} shows that \ac{SGRAND} provides \ac{ML} decisions. However, \ac{SGRAND} requires a dynamic data structure to generate the error patterns rank-ordered by the score depending upon the real-valued reliabilities $\ell^n$. 

\Ac{SRGRAND} uses a one bit quantizer for the reliabilities, i.e., \ac{SRGRAND} sets $w_i=+\infty$ for bits with reliability $\ell_i$ higher than a threshold, tagging them as being perfectly reliable, and $w_i=1$ for those below the threshold, resulting in error patterns following increasing Hamming weight within the region of unreliable bits~\cite{Duffy22}, i.e.,
\begin{equation}
  w_i=\left\{
    \begin{aligned}
      1, &\text{ if } \ell_i < \delta\\
      +\infty, &\text{ otherwise}.
    \end{aligned}\right.
\end{equation}

\subsection{ORBGRAND}\label{sec:ORB}
Figure~\ref{fig:orb_linear} shows that the rank-ordered reliabilities are increasing almost linearly at low to moderate \ac{SNR} regime. Based on this observation, the basic version of \ac{ORBGRAND}~\cite{duffy2021ordered} considers the received bits rank-ordered in increasing reliability and their weights are increasing linearly, i.e., 
\begin{align}\label{eq:orb_weight}
    w_i=r,\quad \text{$\ell_i$ is the $r$-th smallest element in $\ell^n$}.
\end{align}
\Ac{ORBGRAND} sorts the reliabilities $\ell^n$ and set the weights $w^n$ to its rank orders $r\in\left\{1,2,\dots,n\right\}$, i.e., \ac{ORBGRAND} uses a $\lceil\log_2(n)\rceil$ bits input-related statistical model-based quantizer. Then error pattern generation could be solved by determining distinct integer partitions. \Ac{ORBGRAND}'s advantage is that, once ranking is complete, pattern generation can be done on the fly with essentially no memory. \Ac{ORBGRAND} provides near-\ac{ML} decoding for \acp{BLER} greater than $10^{-4}$, but it is less precise at high \ac{SNR}. To overcome this problem, a multi-line \ac{ORBGRAND} with a more sophisticated statistical model is proposed in~\cite{duffy22ORBGRAND}.

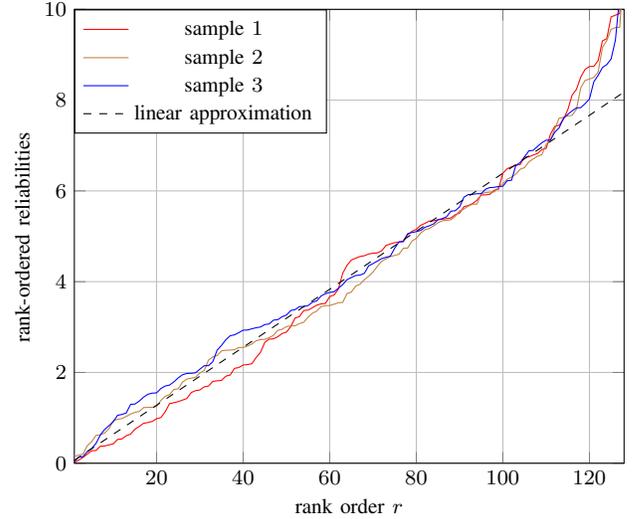
\begin{figure}[t]
	\centering
	\footnotesize
	\begin{tikzpicture}[scale=1]
\begin{axis}[
legend style={at={(0,1)},anchor= north west},
ymin=0,
ymax=10,
width=3.5in,
height=3.0in,
grid=both,
xmin = 1,
xmax = 128,
xlabel = rank order $r$,
ylabel = {rank-ordered reliabilities},
]

\addplot[red]
table[x=snr,y=fer]{snr fer
1 0.014975
2 0.067904
3 0.12906
4 0.20152
5 0.27591
6 0.28069
7 0.3743
8 0.3791
9 0.40448
10 0.42729
11 0.52522
12 0.53462
13 0.60667
14 0.6395
15 0.75146
16 0.81145
17 0.87951
18 0.89729
19 0.92114
20 0.98061
21 1.0064
22 1.1344
23 1.3131
24 1.334
25 1.3579
26 1.3912
27 1.4266
28 1.5567
29 1.5988
30 1.6137
31 1.6832
32 1.7005
33 1.8045
34 1.8162
35 1.8244
36 1.9183
37 1.9427
38 2.0901
39 2.1085
40 2.1636
41 2.1735
42 2.1917
43 2.3291
44 2.441
45 2.6643
46 2.734
47 2.7399
48 2.7863
49 2.8589
50 2.8884
51 3.0113
52 3.1995
53 3.2304
54 3.3654
55 3.3735
56 3.4259
57 3.4833
58 3.5105
59 3.5198
60 3.6718
61 3.7058
62 3.8349
63 4.1846
64 4.3609
65 4.4774
66 4.5181
67 4.5552
68 4.572
69 4.6102
70 4.6276
71 4.6335
72 4.6809
73 4.7975
74 4.8317
75 4.8606
76 4.8731
77 4.8979
78 5.0325
79 5.1102
80 5.1543
81 5.2414
82 5.2859
83 5.3369
84 5.3415
85 5.3653
86 5.3749
87 5.3904
88 5.403
89 5.4761
90 5.5121
91 5.6518
92 5.6781
93 5.7202
94 5.7922
95 5.9092
96 5.9122
97 5.9249
98 6.0144
99 6.0387
100 6.3888
101 6.4895
102 6.5316
103 6.5625
104 6.5812
105 6.7452
106 6.7505
107 6.7907
108 6.8239
109 6.9104
110 6.9378
111 7.2513
112 7.4238
113 7.4466
114 7.5785
115 7.7899
116 8.1112
117 8.2388
118 8.5107
119 8.6787
120 8.7421
121 8.7439
122 8.8699
123 9.3059
124 9.3532
125 9.8335
126 9.8719
127 9.9098
128 12.604

};\addlegendentry{sample $1$}

\addplot[brown]
table[x=snr,y=fer]{snr fer
1 0.15161
2 0.17977
3 0.20004
4 0.38773
5 0.47311
6 0.61269
7 0.62864
8 0.66598
9 0.76012
10 0.93625
11 0.96677
12 0.98504
13 1.0316
14 1.0839
15 1.1215
16 1.1403
17 1.2259
18 1.2305
19 1.2322
20 1.2494
21 1.3826
22 1.4744
23 1.5131
24 1.6254
25 1.6403
26 1.7874
27 1.8536
28 1.8742
29 1.8955
30 1.9872
31 2.0425
32 2.2724
33 2.3447
34 2.3725
35 2.4753
36 2.4899
37 2.4924
38 2.5023
39 2.5497
40 2.5518
41 2.5803
42 2.6261
43 2.6976
44 2.7228
45 2.7297
46 2.7662
47 2.8183
48 2.9288
49 2.9515
50 3.0112
51 3.0268
52 3.0288
53 3.077
54 3.106
55 3.2181
56 3.3044
57 3.353
58 3.4499
59 3.4717
60 3.4751
61 3.5035
62 3.5363
63 3.5439
64 3.7371
65 3.7727
66 3.9024
67 3.945
68 4.0433
69 4.1213
70 4.2113
71 4.3388
72 4.4499
73 4.5033
74 4.554
75 4.5784
76 4.6045
77 4.737
78 4.7401
79 4.8957
80 4.9531
81 5.0699
82 5.1445
83 5.1721
84 5.2285
85 5.2985
86 5.3409
87 5.3508
88 5.3644
89 5.4411
90 5.5629
91 5.6001
92 5.665
93 5.7076
94 5.7098
95 5.9382
96 5.9628
97 5.9733
98 5.9774
99 6.0759
100 6.1594
101 6.2567
102 6.3041
103 6.3712
104 6.4778
105 6.513
106 6.6396
107 6.705
108 6.7632
109 6.7961
110 7.0242
111 7.1282
112 7.3677
113 7.6016
114 7.6225
115 7.6537
116 7.6648
117 7.7984
118 8.2733
119 8.4342
120 8.4641
121 8.5008
122 8.6149
123 9.1617
124 9.3129
125 9.5588
126 9.5982
127 9.6051
128 10.994

};\addlegendentry{sample $2$}

\addplot[blue]
table[x=snr,y=fer]{snr fer
1 0.033887
2 0.13719
3 0.14356
4 0.25756
5 0.32588
6 0.43487
7 0.62136
8 0.74348
9 0.82664
10 0.91591
11 1.0565
12 1.0746
13 1.1085
14 1.2972
15 1.304
16 1.3699
17 1.4564
18 1.5197
19 1.5444
20 1.55
21 1.6462
22 1.7019
23 1.722
24 1.7929
25 1.8812
26 1.9499
27 1.9781
28 1.9818
29 1.9934
30 2.0716
31 2.1484
32 2.1508
33 2.239
34 2.4934
35 2.5965
36 2.7208
37 2.8085
38 2.8294
39 2.8773
40 2.9323
41 2.9344
42 2.9622
43 2.9851
44 3.0013
45 3.0555
46 3.0687
47 3.1527
48 3.1628
49 3.228
50 3.2648
51 3.3644
52 3.3848
53 3.4504
54 3.4541
55 3.4741
56 3.5483
57 3.5761
58 3.686
59 3.7338
60 3.7644
61 3.7782
62 3.8563
63 3.9349
64 4.0462
65 4.0921
66 4.1372
67 4.1456
68 4.1838
69 4.3507
70 4.3989
71 4.4533
72 4.5084
73 4.5377
74 4.5629
75 4.7364
76 4.8762
77 4.8804
78 5.0574
79 5.0836
80 5.1033
81 5.1286
82 5.1986
83 5.2315
84 5.2729
85 5.3656
86 5.3889
87 5.4311
88 5.5542
89 5.5633
90 5.6464
91 5.8616
92 5.9192
93 5.9246
94 5.9375
95 5.9399
96 6.0383
97 6.0736
98 6.0784
99 6.0965
100 6.1026
101 6.2226
102 6.2348
103 6.5623
104 6.6242
105 6.7509
106 6.8822
107 6.8938
108 6.9806
109 7.0618
110 7.1178
111 7.1195
112 7.2955
113 7.3892
114 7.6019
115 7.6876
116 7.7638
117 7.829
118 7.8347
119 7.9228
120 8.0256
121 8.402
122 8.544
123 8.7183
124 8.7823
125 8.9063
126 9.3238
127 10.315
128 10.478

};\addlegendentry{sample $3$}

\addplot[black,dashed]
table[x=snr,y=fer]{snr fer
1 0.063888
2 0.12778
3 0.19166
4 0.25555
5 0.31944
6 0.38333
7 0.44722
8 0.5111
9 0.57499
10 0.63888
11 0.70277
12 0.76666
13 0.83054
14 0.89443
15 0.95832
16 1.0222
17 1.0861
18 1.15
19 1.2139
20 1.2778
21 1.3416
22 1.4055
23 1.4694
24 1.5333
25 1.5972
26 1.6611
27 1.725
28 1.7889
29 1.8528
30 1.9166
31 1.9805
32 2.0444
33 2.1083
34 2.1722
35 2.2361
36 2.3
37 2.3639
38 2.4277
39 2.4916
40 2.5555
41 2.6194
42 2.6833
43 2.7472
44 2.8111
45 2.875
46 2.9388
47 3.0027
48 3.0666
49 3.1305
50 3.1944
51 3.2583
52 3.3222
53 3.3861
54 3.45
55 3.5138
56 3.5777
57 3.6416
58 3.7055
59 3.7694
60 3.8333
61 3.8972
62 3.9611
63 4.0249
64 4.0888
65 4.1527
66 4.2166
67 4.2805
68 4.3444
69 4.4083
70 4.4722
71 4.536
72 4.5999
73 4.6638
74 4.7277
75 4.7916
76 4.8555
77 4.9194
78 4.9833
79 5.0472
80 5.111
81 5.1749
82 5.2388
83 5.3027
84 5.3666
85 5.4305
86 5.4944
87 5.5583
88 5.6221
89 5.686
90 5.7499
91 5.8138
92 5.8777
93 5.9416
94 6.0055
95 6.0694
96 6.1332
97 6.1971
98 6.261
99 6.3249
100 6.3888
101 6.4527
102 6.5166
103 6.5805
104 6.6444
105 6.7082
106 6.7721
107 6.836
108 6.8999
109 6.9638
110 7.0277
111 7.0916
112 7.1555
113 7.2193
114 7.2832
115 7.3471
116 7.411
117 7.4749
118 7.5388
119 7.6027
120 7.6666
121 7.7304
122 7.7943
123 7.8582
124 7.9221
125 7.986
126 8.0499
127 8.1138
128 8.1777

};\addlegendentry{linear approximation}

\end{axis}
\end{tikzpicture}
	\caption{Random samples of rank-ordered reliability at $E_s/N_0=3~\text{dB}$ for $n=128$.}
	\label{fig:orb_linear}
\end{figure}

\section{Achievable rate of ORBGRAND}\label{sec:orbgrand}
In this section, the achievable rate of \ac{ORBGRAND} is discussed. In contrast to the approach taken in~\cite{liu2022orbgrand}, our analysis is based on order statistics and mismatched
decoding theory.

\subsection{Order statistics of the reliabilities}\label{sec:OrderStatistic}

Let $\ell_{\left(r\right)}$ be the $r$-th smallest element in vector $\ell^n$ and $L_{\left(r\right)}$ be the $r$-th order statistic of the reliability samples $L^n$. Since $L_i=|\tau(Y_i)|$, the reliabilities $L_i$ are independent and identically folded Gaussian distributed, i.e.,

\begin{equation}
L_i \overset{\text{i.i.d.}}{\sim} f_{L}(a) = p_\text{FG}\left(a\left|\frac{2}{\sigma^2},\frac{4}{\sigma^2}\right.\right)
\end{equation}
where $p_\text{FG}\left(x\left|\mu,\sigma^2\right.\right)$ denotes the PDF of a folded Gaussian distributed RV, i.e.,
\begin{equation}
  p_\text{FG}\left(a\left|\mu,\sigma^2\right.\right)\hspace{-3pt}=\hspace{-3pt}\left\{
    \begin{aligned}
      p_\text{G}\left(a\left|\mu,\sigma^2\right.\right) +p_\text{G}\left(-a\left|\mu,\sigma^2\right.\right),& \text{~if~} a \geq 0\\
      0, & \text{~otherwise}
    \end{aligned}\right.
\end{equation}
where $p_\text{G}\left(x\left|\mu,\sigma^2\right.\right)$ denotes the \ac{PDF} of a Gaussian \ac{RV}.
We have then the distribution of $L_{\left(r\right)}$
\begin{align}\label{eq:orderstat}
    f_{L_{(r)}}(a) = \frac{n!}{(r-1)!(n-r)!} f_{L}(a) F_{L}(a)^{r-1}\left(1-F_{L}(a)\right)^{n-r}.
\end{align}

As mentioned in Section~\ref{sec:GRANDs}, basic \ac{ORBGRAND} does not provide \ac{ML} decisions in general, since \ac{ORBGRAND} uses an input-related model-based quantizer, which maps the $r$-th smallest reliability $\ell_{(r)}$ to its rank order $r$, i.e.,
\begin{align}
    \ell_{(r)}\mapsto r.
\end{align}
We separate this mapping in two parts,
\begin{align}
    \ell_{(r)}\mapsto\text{E}\left[L_{(r)}\right] \text{~and~} \text{E}\left[L_{(r)}\right]\mapsto r.
\end{align}

We first map the reliability $\ell_{(r)}$ to its expectation $\text{E}\left[L_{(r)}\right]$. Since the reliabilities $\ell^n$ are \ac{i.i.d.} distributed with finite variance, the variance of the order statistics approaches zero asymptotically in the number of samples~\cite{boucheron2012concentration}, i.e.,
\begin{align}
    \lim_{n \to +\infty} \text{Var}\left[L_{(r)}\right]=0.
\end{align}
An example for finite code length is shown in Figure~\ref{fig:AvgEL}. As a result, $\ell_{(r)}\mapsto\text{E}\left[L_{(r)}\right]$ introduces only very limited performance loss. In~\cite[Sec.~3.D]{duffy22ORBGRAND}, a multi-line \ac{ORBGRAND} is proposed. The performance loss of multi-line \ac{ORBGRAND} vanishes if a precise enough curve fitting of $\text{E}\left[L_{(r)}\right]$ is used as weight $w^n$. 

The expectation $\text{E}\left[L_{(r)}\right]$ is then mapped to the rank order $r$, i.e., \ac{ORBGRAND} considers the received bits rank-ordered in increasing reliability and their weights are increasing linearly, which is only precise at moderate \ac{SNR}. An example is shown in Figure~\ref{fig:OrderStat}, the expectations are calculated using eq.~\eqref{eq:orderstat}. To evaluate the loss introduced by $\text{E}\left[L_{(r)}\right]\mapsto r$, we define a continuous curve fitting $\lambda(r)$ for the discrete function $\text{E}\left[L_{(r)}\right]$. By construction, we have that
    \begin{align}
        \lambda^{-1}\left(\text{E}\left[L_{(r)}\right]\right) = r.
    \end{align}
Thus, $\lambda^{-1}\left(\ell_{(r)}\right)$ is an approximation of $r$, i.e.,
    \begin{align}\label{eq:approx_orb}
        \lambda^{-1}\left(\ell_{(r)}\right) \approx r.
    \end{align}
An example for a realization $y^n$ at $8$~dB is shown in Figure~\ref{fig:ORBMetric}. We observe that the approximation in eq.~\eqref{eq:approx_orb} is pretty precise, i.e., by using the function $\lambda^{-1}(\cdot)$, we remove the statistical information of $\text{E}\left[L_{(r)}\right]$ and treat the reliabilities increasing almost linearly as basic \ac{ORBGRAND}. 

\subsection{Achievable rates and mismatched decoding}\label{sec:achievablerate}
Consider a \ac{BMS} channel $f_{Y|C}$ with codewords $C^n$ uniformly distributed in a codebook $\mathcal{C}$. The highest rate that can be supported, in the Shannon sense, is governed by mutual information, $\text{I}\left(C;Y\right)$. An \ac{ML} decoder eq.~\eqref{eq:MLD} uses the \ac{LLR} $\tau(Y)$ as decoding input and achieves the mutual information.

In contrast to \ac{ML} decoding, a mismatched decoder~\cite[Sec. 8.2]{bocherer2018principles} uses a different function $M(Y)$ from $\tau(Y)$ as the decoding input providing the soft information about bit $C$, i.e.,
\begin{align}
    \hat{c}^n =\argmax_{c^n\in\mathcal{C}} \sum_{i=1}^n (1-2c_i) \cdot M\left(y_i\right).\label{eq:misD}
\end{align}
Defining
\begin{align}
    &R(Y,C,M(Y),s,r)= \\
    &\text{E} \left[ \log_2 \frac{e^{s{M(Y)}(1-2C)}r(C)}{\sum_{c\in\{0,1\}}P_C(c)e^{s{M(Y)}(1-2c)}r(c) } \right].
\end{align}
where $s\geq 0$, and where $r:\{0,1\}\rightarrow \mathbb{R}$ is a real-valued function with finite expectation $\text{E}\left[r(C)\right]<\infty$. By \cite{ganti2000mismatched}, it is known that the block error probability of mismatched decoding approaches zero for $n$ approaching infinity if
\begin{align}
   \frac{k}{n}<R\left(Y,C,M(Y),s,r\right).
\end{align}

In general, we have $R\left(Y,C,M(Y),s,r\right)\leq \text{I}\left(C;Y\right)$ with equality if $M(Y)=\tau(Y)$, $s=1$ and $r(\cdot)=1$. Fixing $r(\cdot) = 1$~\cite{bocherer2017efficient} and maximizing over $s$ yields the generalized mutual information~\cite{kaplan1993information}. Maximizing over $r,s$ yields the highest achievable rate, the LM-rate~\cite{ganti2000mismatched,bocherer2018principles}. For a \ac{BMS} channel, we have
\begin{align}
    &R_\text{LM}(Y,C,M(Y))=\max_{r,s} R\left(Y,C,M,s,r\right)\\
    &=\max_{r,s} \text{E} \left[ \log_2 \frac{2e^{s{M(Y)}(1-2C)}r(C)}{e^{s{M(Y)}}r(0)+e^{-s{M(Y)}}r(1) } \right].
\end{align}
Note that we have $R_\text{LM}(Y,C,\tau(Y))=\text{I}\left(C;Y\right)$.

We now define a function $\nu (\cdot)$,
    \begin{align}
        \nu(y)\triangleq \text{sign}(y)\cdot\lambda^{-1}\left(\frac{2|y|}{\sigma^2}\right)
    \end{align}
The LM-rate of a basic \ac{ORBGRAND} is given by
    \begin{align}\label{eq:LMrateORB}
        R_\text{LM}(Y,C,\nu(Y)).
    \end{align}
An example for $n=128$ is shown in Figure~\ref{fig:ORBLoss}, where the loss is defined by
    \begin{align}
        \left[E_s/N_0\right]_{\text{dB}} - \left[C_\text{BPSK}^{-1}( R_\text{LM}(Y,C,\nu(Y)) )\right]_{\text{dB}}.
    \end{align}
Note that eq.~\eqref{eq:LMrateORB} is \emph{not} an achievable rate for finite code length $n$. The code length is only related to the statistical model $\text{E}\left[L_{(r)}\right]$, and consequently related to $\nu(\cdot)$.

We have following conclusions about the performance of \ac{ORBGRAND}: consistent with the results in ~\cite{liu2022orbgrand}, since the rank-ordered weights are increasing almost linearly at low to moderate \ac{SNR} regimes, the basic version of \ac{ORBGRAND} is an effective approach; n the high \ac{SNR} regime, multi-line \ac{ORBGRAND}~\cite[Sec.~3.D]{duffy22ORBGRAND} exhibits a negligible loss relatively to capacity-achieving \ac{SGRAND} if a precise enough curve fitting is used for multi-line \ac{ORBGRAND}.


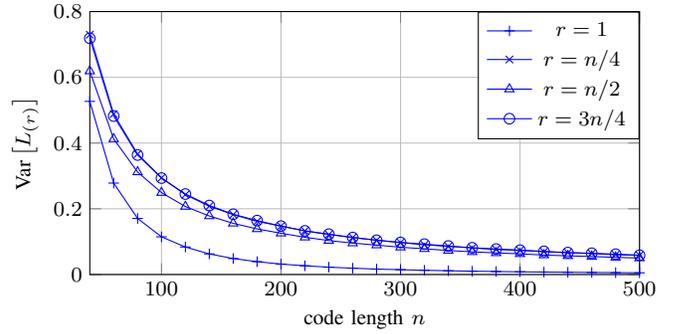
\begin{figure}[t]
	\centering
	\footnotesize
	\begin{tikzpicture}[scale=1]
\begin{axis}[
legend style={at={(1,1)},anchor= north east},
ymin=0,
ymax=0.8,
width=3.5in,
height=2.0in,
grid=both,
xmin = 40,
xmax = 500,
xlabel = code length $n$,
ylabel = {$\text{Var}\left[L_{(r)}\right]$},
]

\addplot[blue,mark=+]
table[x=snr,y=fer]{snr fer
40 0.52675
60 0.27819
80 0.17056
100 0.11472
120 0.083657
140 0.062688
160 0.048595
180 0.039074
200 0.031723
220 0.026581
240 0.022369
260 0.019267
280 0.016462
300 0.014657
320 0.012873
340 0.011429
360 0.01014
380 0.0090746
400 0.0082201
420 0.0074538
440 0.0068854
460 0.0062634
480 0.0057359
500 0.0052467

};\addlegendentry{$r=1$}

\addplot[blue,mark=x]
table[x=snr,y=fer]{snr fer
40 0.7296
60 0.48823
80 0.36655
100 0.29416
120 0.24289
140 0.20782
160 0.18226
180 0.16177
200 0.14721
220 0.13285
240 0.12163
260 0.11282
280 0.10444
300 0.097283
320 0.090417
340 0.085565
360 0.080992
380 0.076351
400 0.073558
420 0.069324
440 0.066435
460 0.063896
480 0.060129
500 0.058489

};\addlegendentry{$r=n/4$}

\addplot[blue,mark=triangle]
table[x=snr,y=fer]{snr fer
40 0.61893
60 0.41247
80 0.31183
100 0.24909
120 0.20682
140 0.17801
160 0.15574
180 0.13865
200 0.1256
220 0.11338
240 0.10333
260 0.095598
280 0.08954
300 0.083048
320 0.078179
340 0.073098
360 0.068947
380 0.065653
400 0.062943
420 0.059192
440 0.056754
460 0.05478
480 0.052068
500 0.050174

};\addlegendentry{$r=n/2$}

\addplot[blue,mark=o]
table[x=snr,y=fer]{snr fer
40 0.71867
60 0.48207
80 0.36357
100 0.29338
120 0.24484
140 0.20989
160 0.18331
180 0.16406
200 0.14726
220 0.13279
240 0.12275
260 0.11255
280 0.10422
300 0.097576
320 0.09242
340 0.086722
360 0.081311
380 0.077519
400 0.074096
420 0.070428
440 0.066686
460 0.06469
480 0.06184
500 0.058688

};\addlegendentry{$r=3n/4$}

\end{axis}
\end{tikzpicture}
	\caption{Variance of the $r$-th order statistic vs. code length (number of samples) at $E_s/N_0=6$~dB.}
	\label{fig:AvgEL}
\end{figure}

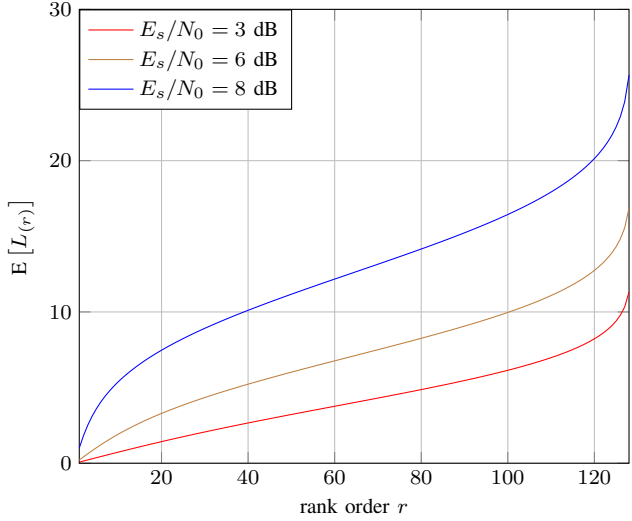
\begin{figure}[t]
	\centering
	\footnotesize
	\begin{tikzpicture}[scale=1]
\begin{axis}[
legend style={at={(0,1)},anchor= north west},
ymin=0,
ymax=30,
width=3.5in,
height=3.0in,
grid=both,
xmin = 1,
xmax = 128,
xlabel = rank order $r$,
ylabel = {$\text{E}\left[L_{(r)}\right]$},
]

\addplot[red]
table[x=snr,y=fer]{snr fer
1 0.073939
2 0.14863
3 0.22247
4 0.29627
5 0.3699
6 0.44334
7 0.51652
8 0.58931
9 0.66158
10 0.73383
11 0.80519
12 0.87622
13 0.94714
14 1.0171
15 1.0868
16 1.1559
17 1.2243
18 1.2923
19 1.3594
20 1.4267
21 1.4932
22 1.5594
23 1.6251
24 1.6898
25 1.7543
26 1.8181
27 1.8811
28 1.9438
29 2.0061
30 2.068
31 2.1294
32 2.1901
33 2.2502
34 2.3099
35 2.3695
36 2.4288
37 2.4875
38 2.5457
39 2.6039
40 2.6615
41 2.7191
42 2.776
43 2.8324
44 2.8888
45 2.9449
46 3.0007
47 3.0561
48 3.1116
49 3.1667
50 3.222
51 3.2769
52 3.3314
53 3.3861
54 3.4406
55 3.4949
56 3.549
57 3.6027
58 3.6568
59 3.7107
60 3.7645
61 3.8188
62 3.8731
63 3.9274
64 3.9817
65 4.0358
66 4.0903
67 4.1445
68 4.1991
69 4.2537
70 4.3082
71 4.3628
72 4.418
73 4.4732
74 4.5284
75 4.5842
76 4.6405
77 4.6966
78 4.7532
79 4.8104
80 4.868
81 4.9261
82 4.9841
83 5.0426
84 5.1015
85 5.161
86 5.2209
87 5.2815
88 5.3424
89 5.4044
90 5.4672
91 5.5307
92 5.5947
93 5.6597
94 5.7258
95 5.7927
96 5.8604
97 5.9297
98 5.9999
99 6.0717
100 6.1447
101 6.2191
102 6.295
103 6.3725
104 6.4519
105 6.5335
106 6.6174
107 6.7031
108 6.7922
109 6.8835
110 6.9788
111 7.0773
112 7.1796
113 7.2866
114 7.3984
115 7.5154
116 7.6389
117 7.7697
118 7.9107
119 8.0612
120 8.2242
121 8.4024
122 8.6005
123 8.8245
124 9.0815
125 9.3918
126 9.7838
127 10.335
128 11.319

};\addlegendentry{$E_s/N_0=3$ dB}

\addplot[brown]
table[x=snr,y=fer]{snr fer
1 0.21382
2 0.42506
3 0.63316
4 0.8354
5 1.0323
6 1.2227
7 1.406
8 1.5832
9 1.7538
10 1.9181
11 2.0774
12 2.2303
13 2.378
14 2.5215
15 2.6599
16 2.7942
17 2.9239
18 3.0498
19 3.1721
20 3.2912
21 3.4074
22 3.5212
23 3.631
24 3.7391
25 3.8452
26 3.9487
27 4.0502
28 4.1493
29 4.2474
30 4.3426
31 4.4369
32 4.5295
33 4.6203
34 4.7102
35 4.7992
36 4.8869
37 4.9733
38 5.0589
39 5.1434
40 5.2267
41 5.3088
42 5.3905
43 5.4712
44 5.5512
45 5.6302
46 5.7089
47 5.7874
48 5.8657
49 5.943
50 6.0193
51 6.0957
52 6.1716
53 6.2471
54 6.3225
55 6.3976
56 6.4723
57 6.5465
58 6.6203
59 6.6942
60 6.7679
61 6.8416
62 6.9156
63 6.9891
64 7.0629
65 7.1367
66 7.2105
67 7.2839
68 7.3575
69 7.431
70 7.505
71 7.579
72 7.6531
73 7.7278
74 7.8024
75 7.8773
76 7.953
77 8.0284
78 8.1048
79 8.1818
80 8.2589
81 8.3366
82 8.4145
83 8.4933
84 8.5722
85 8.6522
86 8.7327
87 8.8141
88 8.896
89 8.9793
90 9.0632
91 9.148
92 9.2341
93 9.3209
94 9.409
95 9.4985
96 9.5887
97 9.681
98 9.775
99 9.871
100 9.9685
101 10.068
102 10.169
103 10.273
104 10.38
105 10.489
106 10.601
107 10.716
108 10.835
109 10.957
110 11.085
111 11.217
112 11.353
113 11.496
114 11.646
115 11.803
116 11.968
117 12.142
118 12.329
119 12.528
120 12.746
121 12.983
122 13.248
123 13.546
124 13.891
125 14.303
126 14.828
127 15.558
128 16.878

};\addlegendentry{$E_s/N_0=6$ dB}

\addplot[blue]
table[x=snr,y=fer]{snr fer
1 1.0046
2 1.8429
3 2.5402
4 3.1252
5 3.6236
6 4.0584
7 4.4418
8 4.7841
9 5.0968
10 5.3837
11 5.6492
12 5.8959
13 6.1284
14 6.3466
15 6.556
16 6.7545
17 6.9448
18 7.1263
19 7.3017
20 7.4699
21 7.6332
22 7.791
23 7.9444
24 8.0939
25 8.2396
26 8.3812
27 8.5198
28 8.6559
29 8.789
30 8.9186
31 9.0455
32 9.1702
33 9.2935
34 9.4145
35 9.5342
36 9.6524
37 9.7682
38 9.8827
39 9.9961
40 10.108
41 10.218
42 10.328
43 10.436
44 10.544
45 10.651
46 10.756
47 10.861
48 10.965
49 11.069
50 11.172
51 11.274
52 11.375
53 11.476
54 11.577
55 11.676
56 11.777
57 11.876
58 11.975
59 12.074
60 12.173
61 12.272
62 12.37
63 12.468
64 12.566
65 12.664
66 12.762
67 12.861
68 12.959
69 13.057
70 13.155
71 13.255
72 13.354
73 13.453
74 13.553
75 13.653
76 13.753
77 13.855
78 13.955
79 14.058
80 14.16
81 14.263
82 14.367
83 14.472
84 14.577
85 14.684
86 14.791
87 14.9
88 15.009
89 15.12
90 15.232
91 15.345
92 15.46
93 15.576
94 15.694
95 15.813
96 15.934
97 16.057
98 16.183
99 16.311
100 16.441
101 16.574
102 16.709
103 16.848
104 16.99
105 17.135
106 17.284
107 17.438
108 17.595
109 17.759
110 17.928
111 18.103
112 18.285
113 18.475
114 18.674
115 18.882
116 19.103
117 19.335
118 19.583
119 19.85
120 20.139
121 20.454
122 20.807
123 21.204
124 21.663
125 22.216
126 22.911
127 23.89
128 25.647

};\addlegendentry{$E_s/N_0=8$ dB}

\end{axis}
\end{tikzpicture}
	\caption{$\text{E}\left[L_{(r)}\right]$ for $n=128$ at $Es/N_0=\left\{3,6,8\right\}$~dB.}
	\label{fig:OrderStat}
\end{figure}

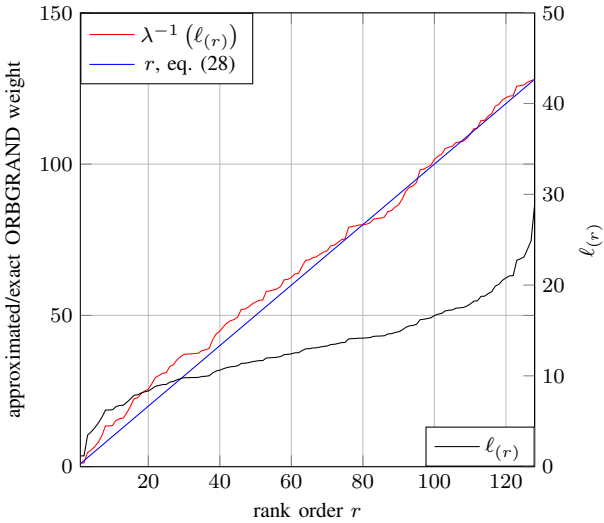
\begin{figure}[t]
	\centering
	\footnotesize
	\begin{tikzpicture}[scale=1]
\begin{axis}[
legend style={at={(0,1)},anchor= north west},
ymin=0,
ymax=150,
width=3in,
height=3.0in,
grid=both,
xmin = 1,
xmax = 128,
xlabel = rank order $r$,
ylabel = {approximated/exact ORBGRAND weight},
]

\addplot[red]
table[x=snr,y=fer]{snr fer
1 1.197
2 1.2104
3 4.6734
4 5.4732
5 6.5735
6 8.07
7 10.186
8 13.436
9 13.459
10 13.604
11 15.312
12 15.876
13 16.048
14 17.82
15 19.905
16 22.446
17 22.747
18 24.148
19 25.014
20 25.545
21 27.44
22 29.463
23 30.144
24 30.885
25 31.003
26 33.109
27 33.697
28 35.283
29 36.313
30 37.149
31 37.196
32 37.31
33 37.402
34 37.499
35 38.374
36 38.585
37 39.099
38 41.937
39 43.928
40 44.784
41 46.272
42 47.481
43 48.26
44 48.628
45 49.447
46 51.975
47 52.074
48 52.762
49 53.736
50 54.414
51 54.984
52 55.091
53 57.897
54 58.104
55 58.42
56 58.782
57 59.552
58 61.759
59 61.912
60 62.564
61 63.717
62 63.851
63 66.257
64 68.213
65 68.332
66 69.118
67 69.336
68 70.194
69 70.913
70 71.352
71 72.974
72 73.239
73 74.074
74 75.033
75 75.279
76 79.122
77 79.354
78 79.574
79 79.901
80 79.908
81 80.164
82 80.578
83 81.829
84 82.045
85 82.152
86 82.498
87 84.33
88 84.631
89 85.817
90 86.574
91 88.487
92 90.996
93 92.375
94 92.782
95 93.939
96 98.088
97 98.337
98 98.992
99 99.633
100 101.54
101 102.61
102 103.23
103 105.11
104 105.59
105 105.87
106 107.11
107 107.31
108 107.55
109 108.46
110 109.92
111 111.68
112 111.89
113 114.36
114 114.47
115 115.8
116 116.7
117 119.13
118 119.77
119 121.19
120 121.95
121 122.5
122 122.56
123 125.73
124 125.97
125 126.17
126 127.04
127 127.56
128 128

};\addlegendentry{$\lambda^{-1}\left(\ell_{(r)}\right)$}

\addplot[blue]
table[x=snr,y=fer]{snr fer
1 1
2 2
3 3
4 4
5 5
6 6
7 7
8 8
9 9
10 10
11 11
12 12
13 13
14 14
15 15
16 16
17 17
18 18
19 19
20 20
21 21
22 22
23 23
24 24
25 25
26 26
27 27
28 28
29 29
30 30
31 31
32 32
33 33
34 34
35 35
36 36
37 37
38 38
39 39
40 40
41 41
42 42
43 43
44 44
45 45
46 46
47 47
48 48
49 49
50 50
51 51
52 52
53 53
54 54
55 55
56 56
57 57
58 58
59 59
60 60
61 61
62 62
63 63
64 64
65 65
66 66
67 67
68 68
69 69
70 70
71 71
72 72
73 73
74 74
75 75
76 76
77 77
78 78
79 79
80 80
81 81
82 82
83 83
84 84
85 85
86 86
87 87
88 88
89 89
90 90
91 91
92 92
93 93
94 94
95 95
96 96
97 97
98 98
99 99
100 100
101 101
102 102
103 103
104 104
105 105
106 106
107 107
108 108
109 109
110 110
111 111
112 112
113 113
114 114
115 115
116 116
117 117
118 118
119 119
120 120
121 121
122 122
123 123
124 124
125 125
126 126
127 127
128 128

};\addlegendentry{$r$, eq.~\eqref{eq:orb_weight}}

\end{axis}

\begin{axis}[
legend style={at={(1,0)},anchor= south east},
ymin=0,
ymax=50,
hide x axis,
axis y line*=right,
ylabel near ticks,
width=3in,
height=3.0in,
xmin = 1,
xmax = 128,
xlabel = rank order $r$,
ylabel = {$\ell_{(r)}$},
]

\addplot[black]
table[x=snr,y=fer]{snr fer
1 1.1757
2 1.187
3 3.4704
4 3.8376
5 4.2844
6 4.8134
7 5.4413
8 6.2322
9 6.2372
10 6.269
11 6.6259
12 6.7375
13 6.7711
14 7.1015
15 7.4617
16 7.8659
17 7.9121
18 8.1218
19 8.2478
20 8.3232
21 8.5863
22 8.8549
23 8.9428
24 9.0374
25 9.0524
26 9.313
27 9.3844
28 9.573
29 9.6938
30 9.7909
31 9.7963
32 9.8094
33 9.8199
34 9.8311
35 9.931
36 9.9548
37 10.013
38 10.326
39 10.541
40 10.632
41 10.789
42 10.915
43 10.996
44 11.034
45 11.118
46 11.376
47 11.386
48 11.456
49 11.554
50 11.622
51 11.678
52 11.689
53 11.968
54 11.989
55 12.02
56 12.056
57 12.132
58 12.349
59 12.364
60 12.428
61 12.541
62 12.554
63 12.79
64 12.981
65 12.993
66 13.07
67 13.091
68 13.176
69 13.247
70 13.29
71 13.452
72 13.478
73 13.562
74 13.658
75 13.683
76 14.073
77 14.097
78 14.12
79 14.154
80 14.154
81 14.181
82 14.224
83 14.354
84 14.377
85 14.388
86 14.425
87 14.617
88 14.649
89 14.777
90 14.859
91 15.069
92 15.35
93 15.509
94 15.556
95 15.693
96 16.199
97 16.231
98 16.314
99 16.397
100 16.65
101 16.797
102 16.884
103 17.155
104 17.226
105 17.268
106 17.458
107 17.489
108 17.528
109 17.675
110 17.919
111 18.232
112 18.27
113 18.755
114 18.778
115 19.063
116 19.271
117 19.895
118 20.08
119 20.533
120 20.799
121 21.015
122 21.038
123 22.736
124 22.898
125 23.084
126 23.962
127 24.88
128 28.949

};\addlegendentry{$\ell_{(r)}$}

\end{axis}

\end{tikzpicture}
	\caption{$\lambda^{-1}\left(\ell_{(r)}\right)$ and $r$ for $n=128$ at $Es/N_0=8$~dB.}
	\label{fig:ORBMetric}
\end{figure}

\begin{figure}[t]
	\centering
	\footnotesize
	\begin{tikzpicture}[scale=1]
\begin{axis}[
legend style={at={(1,0)},anchor= south east},
ymin=0,
ymax=0.08,
width=3.5in,
height=2.0in,
ytick={0.01,0.02,...,0.08},
grid=both,
xmin = 2,
xmax = 8,
xlabel = $E_s/N_0$ in dB,
ylabel = {Loss in dB},
]

\addplot[red]
table[x=snr,y=fer]{snr fer
2 0.0027008
2.5 0.0037384
3 0.0063019
3.5 0.010391
4 0.017342
4.5 0.025055
5 0.032196
5.5 0.0401
6 0.048386
6.5 0.055527
7 0.062286
7.5 0.069427
8 0.075806

};

\end{axis}
\end{tikzpicture}
	\caption{Basic ORBGRAND loss for $n=128$.}
	\label{fig:ORBLoss}
\end{figure}
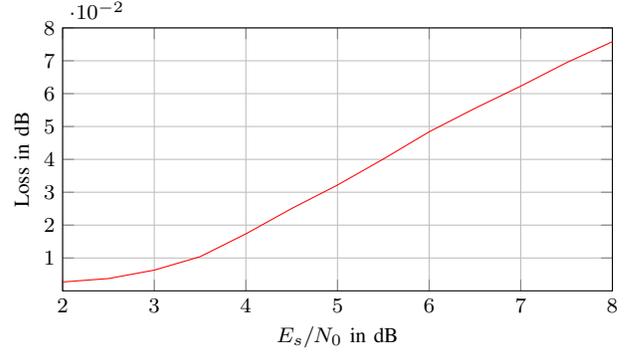

\section{Discretized soft GRAND}\label{sec:qgrand}
As discussed in Section~\ref{sec:GRANDs}, error pattern generation rank-ordered by score $\sum_{i=1}^n e_i\cdot \ell_i$ requires dynamic memory and is hard to implement in hardware efficiently. In contrast to \ac{ORBGRAND}, \ac{DSGRAND} envisages a conventional quantization of the real-valued reliabilities $\ell^n$ into a restricted number of categories determined by a quantization level without the need for received bit reliabilities to be rank-ordered. 

\subsection{Detailed algorithm}\label{sec:dsgrand_detail}
The reliabilities $\ell^n$ are quantized to discrete input weight $w^n$ via
\begin{align}
    w_i = \mu\left(\ell_i\right),~i=1,\dots,n.
\end{align}

The pseudo code of the proposed \ac{DSGRAND} is shown in Algorithm~\ref{alg:qGRAND}. A two-dimensional array $\Lambda$ of size $n\times s_\text{max}$ is required as a \ac{DP} table to store the $\texttt{boolean}$ type elements, where $s_\text{max}$ denotes the maximum score of the error patterns the \ac{DSGRAND} may generate (details are discussed in Section~\ref{sec:Smax}).

For a given score $s$, the error pattern generation is equivalent to a subset sum problem~\cite{kleinberg2006algorithm}, i.e., our target is to find all possible binary vectors $e^n$ such that 
\begin{align}
    \sum_{i=1}^n e_i\cdot w_i = s.
\end{align}
We first check the existence of error patterns (line $6-11$). The element $\Lambda[i,s+1]$ denotes whether we could find a subset in $w^i$ to achieve sum $s$. $\Lambda[i,s+1]$ is marked as true if one of the following conditions is met,
\begin{itemize}
    \item $w_i=s$: $w_i$ is equal to the sum $s$.
    \item $\Lambda[i-1,s+1]=\text{true}$: If there exists a subset in $w^{i-1}$ to achieve sum $s$, the sum $s$ could be achieved if $w_i$ is also allowed to be considered.
    \item $\Lambda[i-1,s-w_i+1]=\text{true}$: If there exists a subset in $w^{i-1}$ to achieve sum $s-w_i$, the sum $s$ is achieved by including $w_i$ in the subset.
\end{itemize}
By reusing the results of score $s-1$, we need $\mathcal{O}(n)$ operations to check the existence of error patterns with score $s$. 

If the existence $\Lambda[n,s+1]$ is marked as true (line $12-13$), the error patterns are generated with help of a stack $\Xi$ to avoid recursive functions. The stack $\Xi$ contains pattern structures $\langle i,j,d,e^n\rangle$, where 
\begin{itemize}
    \item $i$ denotes the row index in $\Lambda$.
    \item $j$ denotes the target score. Note that $j$ is not always equal to $s$.
    \item $e^n$ denotes the current error pattern.
    \item $d$ denotes the difference between the sum $s$ and the score of the current error pattern, i.e., $d=s-\sum_{i=1}^n e_i\cdot w_i$.
\end{itemize}
Firstly, an initial structure $\langle n,s,s-w_n, (0^{n-1},1)\rangle$ is pushed into the stack $\Xi$ (line $15$). Then, we pop pattern structures from the stack and push new pattern structures until it is empty to generate all error patterns with score $s$. If we get a structure $\langle i,j,d,e^n\rangle$ from $\Xi$ (line $17$),
\begin{itemize}
    \item If $d=0$, an error pattern with score $s$ is found. We check whether $\tilde{c}^n\oplus e^n \in\mathcal{C}$. If so, then $e^n$ the estimate of the noise effect $\hat{z}^n$ given the discretized soft information (line $18-22$).
    \item If $\Lambda(i-1,j+1)$ is true, there is a solution with $w_{i-1}$ but without $w_i$ to achieve the sum $j$ (line $24-26$). Thus, we push a new structure into the stack by assuming $e_{i-1}=1$ and $e_{i}=0$ for target score $j$.
    \item If $\Lambda(i-1,j-w_{i}+1)$ is true, there is a solution with $w_{i-1}$ to achieve the sum $j-w_{i}$  (line $27-29$). Thus, we push a new structure into the stack by assuming $e_{i-1}=1$ for target score $j-w_i$.
    \item The error pattern $e^n$ popped from the stack is duplicated (line $23$) only if the both conditions in line $24$ and line $27$ are met.
\end{itemize}
In general, the proposed \ac{DSGRAND} algorithm requires $\mathcal{O}\left(n\cdot S(\hat{z})\right)$ operations to construct the \ac{DP} table and $\mathcal{O}(n_p)$ operations to generate and test the error patterns. 

In contrast to the set $\mathcal{S}$ used in \ac{SGRAND}~\cite{solomon20SGRAND}, the stack $\Xi$ in our implementation is not order-sensitive, i.e., for a given score $s$, every single thread could execute line $17-29$ independently to other threads without maintaining the order of the stack. In our experiments, the size of $\Xi$ never exceeds $12$.

We note that a decoder with discrete input may find multiple codewords with a same likelihood, which is a probability zero event for the continuous case. In our implementation, \ac{DSGRAND} returns the first valid codeword that is found \emph{without} checking for the existence of other codewords with the same score.


\begin{algorithm}[t]
	\SetNoFillComment
	\DontPrintSemicolon
	\SetKwInOut{Input}{Input}\SetKwInOut{Output}{Output}
	\Input{hard decisions $\tilde{c}^n$, weights $w^n$, maximum score $s_\text{max}$}
	\Output{estimates $\hat{c}^n$, number of guesses $n_p$, decoding state $\phi$}
	\BlankLine
	$\phi\leftarrow\text{false}, n_p\leftarrow 1$\\
	\If{$\tilde{c}^n\in\mathcal{C}$}{
	    $\hat{c}^n\leftarrow \tilde{c}^n, \phi\leftarrow\text{true}$\\
	    \Return{}
	}
	\For{$s=0,1,\dots,s_\text{max}$}{
		$\Lambda[1,s+1]\leftarrow (w_1=s)$\\
		\For{$i=2,3,\dots,n$}{
		    $\Lambda[i,s+1]\leftarrow \Lambda[i-1,s+1]\text{~or~}(w_i=s)$\\
		    \If{$\Lambda[i,s+1]=\text{false}$}{
	            \If{$s-w_i\geq 0$}{
	                $\Lambda[i,s+1]\leftarrow \Lambda[i-1,s-w_i+1]$
	            }
	        }
	    }
	    \If{$\Lambda[n,s+1]=\text{false}$}{
	        \Continue
	    }
	    $\Xi\leftarrow\varnothing$\\
	    $\text{push}\left(\Xi, \langle n,s,s-w_n, (0^{n-1},1)\rangle\right)$\\
	    \While{$\Xi\neq\varnothing$}{
		    $\langle i,j,d,e^n \rangle \leftarrow \text{pop}(\Xi)$\\
		    \If{$d=0$}{
	            $n_p\leftarrow n_p+1$\\
	            \If{$\tilde{c}^n\oplus e^n \in\mathcal{C}$}{
    	            $\hat{c}^n\leftarrow \tilde{c}^n\oplus e^n,\phi\leftarrow\text{true}$\\
	                \Return{}
	            }
	        }
	        $\textcolor{red}{\underline{e}^n}\leftarrow \text{copy}(e^n)$\tcp*{if necessary}
	        \If{$i>1,~\Lambda(i-1,j+1)$}{
	           $\textcolor{red}{\underline{e}_{i}}\leftarrow 0,~\textcolor{red}{\underline{e}_{i-1}}\leftarrow 1$\\
	            $\text{push}\left(\Xi, \langle i-1,j,d+w_{i}-w_{i-1}, \textcolor{red}{\underline{e}^n}\rangle\right)$\\
	        }
	        \If{$i>1,~j-w_{i}\geq 0,~\Lambda(i-1,j-w_{i}+1)$}{
	                $e_{i-1}\leftarrow 1$\\
	                $\text{push}\left(\Xi, \langle i-1,j-w_{i},d-w_{i-1}, e^n\rangle\right)$\\
	        }
	    }
	}
	\Return{}
	\caption{$\text{DSGRAND}$}
	\label{alg:qGRAND}
\end{algorithm}

\subsection{Quantizer optimization}
For a system with conventional quantizer, the mutual information is given by $\text{I}\left(C;\tilde{\mu}(Y)\right)$. Replacing the optimal decoder input $\tau(Y^n)$ with $\bar{\mu}(\tau(Y^n))$ can be thought of as using a mismatched decoder~\cite{ganti2000mismatched}. Note that $\tilde{\mu}(\cdot)$ and $\bar{\mu}(\cdot)$ denote the channel output quantizer and \ac{LLR} quantizer, respectively, which are equivalent to the reliability quantizer $\mu(\cdot)$ (see Section~\ref{sec:quantizer}).

In this work, we consider three types of conventional quantizers.
\subsubsection*{Heuristic}
An uniform quantizer is used. The step size $\beta$ of quantization is heuristically chosen via
\begin{align}\label{eq:heuristic}
    \beta = \frac{2}{\sigma^2} \frac{1-\sigma/2}{Q}.
\end{align}
The first term $2/\sigma^2$ normalizes for the increase in reliability with \ac{SNR}, while the second term, $(1-\sigma/2)/Q$ ensures that approximately $30\%$ of the least reliable bits are accurately assigned quantized reliabilities, while the $70\%$ most reliable bits are grouped together, 
\begin{equation}
        b_i = i\beta,\quad i=1,\dots,Q-1.
\end{equation}
\subsubsection*{Uniform quantizer}
An uniform quantizer uses fixed quantization step size $\beta$, while the output values $v^Q$ are not constrained. Review the mutual information and LM-rate of systems with quantized output introduced in Section~\ref{sec:achievablerate}. We know that the mutual information $\text{I}\left(C;\tilde{\mu}(Y)\right)$ is only related to the quantization boundaries $b^{Q-1}$. The LM-rate $R_\text{LM}(Y,C,\bar{\mu}(\tau(Y)))$ is equal to the mutual information only if the output values $v^Q$ are optimized. We first optimize the step size $\beta$ for boundaries $b_i = i\beta,~i=1,\dots,Q-1$ by maximizing the mutual information,
    \begin{align}
        \beta = \argmax_{\beta} \text{I}\left(C;\tilde{\mu}(Y)\right).
    \end{align}
\subsubsection*{Non-uniform quantizer}
A non-uniform quantizer has no constraints on the boundaries and output values. Similar to the uniform quantizer, we optimize the boundaries by
    \begin{align}
        b_1,\dots,b_{Q-1} = \argmax_{b_1,\dots,b_{Q-1}} \text{I}\left(C;\tilde{\mu}(Y)\right).
    \end{align}

For given boundaries $b^{Q-1}$, the optimal output values $v^Q$ is always given by
\begin{align}\label{eq:optimal_v}
     v_i = \log \frac{ \int_{b_{i-1}}^{b_i} f_{\tau(Y)|C}(x|0)\ \text{d}x }{\int_{b_{i-1}}^{b_i} f_{\tau(Y)|C}(x|1)\ \text{d}x},~i=1,\dots,Q.
\end{align}
As mentioned in Section~\ref{sec:dsgrand_detail}, our implementation of \ac{DSGRAND} only accepts positive integers as weights. We now map the optimal output values $v^Q$ to integers. Because of eq.~\eqref{eq:weightscaling}, the output values $v^Q$ are normalized by $\alpha=1/v_1$, such that $v_1=1$. The output values $v^Q$ are then mapped to their nearest integers via round function. The resulting LM-rate is given by
\begin{align}
    R_\text{LM}(Y,C,\bar{\mu}(\tau(Y))).
\end{align}
Note that we have 
\begin{align}
    R_\text{LM}(\tilde{\mu}(Y),C,\bar{\mu}(\tau(Y))) = R_\text{LM}(Y,C,\bar{\mu}(\tau(Y)))
\end{align}
since the \ac{LLR} quantizer $\bar{\mu}(\tau(Y))$ carries the equivalent boundaries information as channel output quantizer $\tilde{\mu}(Y)$. 

Design examples for above mentioned quantizers for \ac{BPSK} modulated codeword over real-valued \ac{AWGN} channels at $4$~dB and $7$~dB are shown in Table~\ref{tab:quantizer4} and Table~\ref{tab:quantizer7}. Their achievable rates are displayed in Figure~\ref{fig:AIR_qGRAND}. 

\begin{table}[t]
    \centering
    \setlength{\tabcolsep}{12pt}
    \renewcommand\arraystretch{1.5}
    \caption{$2$ bits quantizers designed for $E_s/N_0=4~\text{dB}$.}
\begin{tabular}{ |c|c|c|}
\hline
$C_\text{BPSK}=0.7944$& input range & output value  \\
\hline
\hline
\multirow{3}{*}{\shortstack{Non-uniform\\ $R_\text{LM}=0.7884$ \\ \textcolor{red}{$R_\text{LM}=0.7883$} }} & $[0, 1.1352)$ & $1.0000$ \textcolor{red}{$(1)$}\\ 
\cline{2-3}
& $[1.1352, 2.4582)$ & $3.1534$ \textcolor{red}{$(3)$}\\ 
\cline{2-3}
& $[2.4582, 4.3560)$ & $5.8914$ \textcolor{red}{$(6)$}\\
\cline{2-3}
& $[4.3560, +\infty)$ & $10.5689$ \textcolor{red}{$(11)$}\\ 
\hline
\hline
\multirow{3}{*}{\shortstack{Uniform\\ $R_\text{LM}=0.7879$ \\ \textcolor{red}{$R_\text{LM}=0.7878$} }} & $[0, 1.3641)$ & $1.0000$ \textcolor{red}{$(1)$}\\ 
\cline{2-3}
& $[1.3641, 2.7281)$ & $3.0001$ \textcolor{red}{$(3)$}\\ 
\cline{2-3}
& $[2.7281, 4.0922)$ & $5.0003$ \textcolor{red}{$(5)$}\\
\cline{2-3}
& $[4.0922, +\infty)$ & $8.5196$ \textcolor{red}{$(8)$}\\ 
\hline
\hline
\multirow{3}{*}{\shortstack{Heuristic\\ $R_\text{LM}=0.7821$\\ \textcolor{red}{$R_\text{LM}=0.7820$}}} & $[0,0.8597)$ & $1.0000$ \textcolor{red}{$(1)$} \\ 
\cline{2-3}
& $[0.8597, 1.7194)$ & $3.0000$ \textcolor{red}{$(3)$}\\ 
\cline{2-3}
& $[1.7194, 2.5792)$ &$5.0000$ \textcolor{red}{$(5)$}\\
\cline{2-3}
& $[2.5792, +\infty)$ & $10.6525$ \textcolor{red}{$(11)$}\\ 
\hline
\end{tabular}
\vspace{10pt}
    \label{tab:quantizer4}
\end{table}

\begin{table}[t]
    \centering
    \setlength{\tabcolsep}{12pt}
    \renewcommand\arraystretch{1.5}
    \caption{$2$ bits quantizers designed for $E_s/N_0=7~\text{dB}$.}
\begin{tabular}{ |c|c|c|}
\hline
$C_\text{BPSK}=0.9507$&input range & output value  \\
\hline
\hline
\multirow{3}{*}{\shortstack{Non-uniform\\ $R_\text{LM}=0.9486$ \\ \textcolor{red}{$R_\text{LM}=0.9485$} }} & $[0, 1.3878)$ & $1.0000$ \textcolor{red}{$(1)$}\\ 
\cline{2-3}
& $[1.3878, 3.0636)$ & $3.1960$ \textcolor{red}{$(3)$}\\ 
\cline{2-3}
& $[3.0636, 5.6249)$ & $6.1455$ \textcolor{red}{$(6)$}\\
\cline{2-3}
& $[5.6249, +\infty)$ & $11.8671$ \textcolor{red}{$(12)$}\\ 
\hline
\hline
\multirow{3}{*}{\shortstack{Uniform\\ $R_\text{LM}=0.9483$\\ \textcolor{red}{$R_\text{LM}=0.9483$} }} & $[0, 1.7278)$ & $1.0000$ \textcolor{red}{$(1)$}\\ 
\cline{2-3}
& $[1.7278, 3.4557)$ & $3.0000$ \textcolor{red}{$(3)$}\\ 
\cline{2-3}
& $[3.4557, 5.1835)$ & $5.0001$ \textcolor{red}{$(5)$}\\
\cline{2-3}
& $[5.1835, +\infty)$ & $9.1775$ \textcolor{red}{$(9)$}\\ 
\hline
\hline
\multirow{3}{*}{\shortstack{Heuristic\\ $R_\text{LM}=0.9481$\\ \textcolor{red}{$R_\text{LM}=0.9481$} }} & $[0, 1.9463)$ & $1.0000$ \textcolor{red}{$(1)$} \\ 
\cline{2-3}
& $[1.9463, 3.8925)$ & $3.0001$ \textcolor{red}{$(3)$}\\ 
\cline{2-3}
& $[3.8925, 5.8388)$ & $5.0003$ \textcolor{red}{$(5)$}\\
\cline{2-3}
& $[5.8388, +\infty)$ & $8.6992$ \textcolor{red}{$(9)$}\\ 
\hline
\end{tabular}
\vspace{10pt}
    \label{tab:quantizer7}
\end{table}

\begin{figure}[t]
	\centering
	\footnotesize
	\begin{tikzpicture}[scale=1]
\begin{axis}[
legend style={at={(0.5,-0.2)},anchor=north,draw=none,/tikz/every even column/.append style={column sep=1mm},cells={align=left}},
legend columns=2,
legend cell align=left,
ymin=0.95,
ymax=0.9625,
width=3.5in,
height=3.0in,
y tick label style={
                /pgf/number format/.cd,
                    fixed relative,
                    precision=3,
                    zerofill,
                /tikz/.cd,
            },
grid=both,
xmin = 7,
xmax = 7.5,
xlabel = $E_s/N_0$ in dB,
ylabel = {bits/channel use},
]

\addplot[black]
table[x=snr,y=fer]{snr fer
7 0.95068
7.1 0.95379
7.2 0.95677
7.3 0.95961
7.4 0.96233
7.5 0.96491
};\addlegendentry{BPSK capacity}

\addplot[gray, mark=+]
table[x=snr,y=fer]{snr fer
7 0.9487
7.1 0.95181
7.2 0.95483
7.3 0.95774
7.4 0.9605
7.5 0.96314
};\addlegendentry{ORBGRAND}

\addplot[red,dashed]
table[x=snr,y=fer]{snr fer
7 0.94854
7.1 0.95177
7.2 0.95485
7.3 0.9578
7.4 0.96062
7.5 0.9633
};\addlegendentry{non-uniform $q=2$}

\addplot[blue,dashed]
table[x=snr,y=fer]{snr fer
7 0.94828
7.1 0.95151
7.2 0.95461
7.3 0.95756
7.4 0.96038
7.5 0.96307
};\addlegendentry{uniform $q=2$}

\addplot[brown,dashed]
table[x=snr,y=fer]{snr fer
7 0.94814
7.1 0.95133
7.2 0.95439
7.3 0.95731
7.4 0.96009
7.5 0.96274
};\addlegendentry{heuristic $q=2$}

\addplot[red]
table[x=snr,y=fer]{snr fer
7 0.95013
7.1 0.95328
7.2 0.95628
7.3 0.95915
7.4 0.96189
7.5 0.9645
};\addlegendentry{non-uniform $q=3$}

\addplot[blue]
table[x=snr,y=fer]{snr fer
7 0.94999
7.1 0.95314
7.2 0.95615
7.3 0.95903
7.4 0.96177
7.5 0.96439

};\addlegendentry{uniform $q=3$}

\addplot[brown]
table[x=snr,y=fer]{snr fer
7 0.94999
7.1 0.95313
7.2 0.95614
7.3 0.95901
7.4 0.96175
7.5 0.96435
};\addlegendentry{heuristic $q=3$}

\addplot[blue,dash dot]
table[x=snr,y=fer]{snr fer
7 0.94133
7.1 0.94489
7.2 0.94831
7.3 0.95158
7.4 0.95472
7.5 0.95771

};\addlegendentry{uniform/non-uniform $q=1$}

\addplot[brown,dash dot]
table[x=snr,y=fer]{snr fer
7 0.94036
7.1 0.9438
7.2 0.94709
7.3 0.95025
7.4 0.95327
7.5 0.95616
};\addlegendentry{heuristic $q=1$}

\end{axis}
\end{tikzpicture}
	\caption{LM-rates of the quantizers. The output values $v^Q$ are mapped to their nearest integers.}
	\label{fig:AIR_qGRAND}
\end{figure}
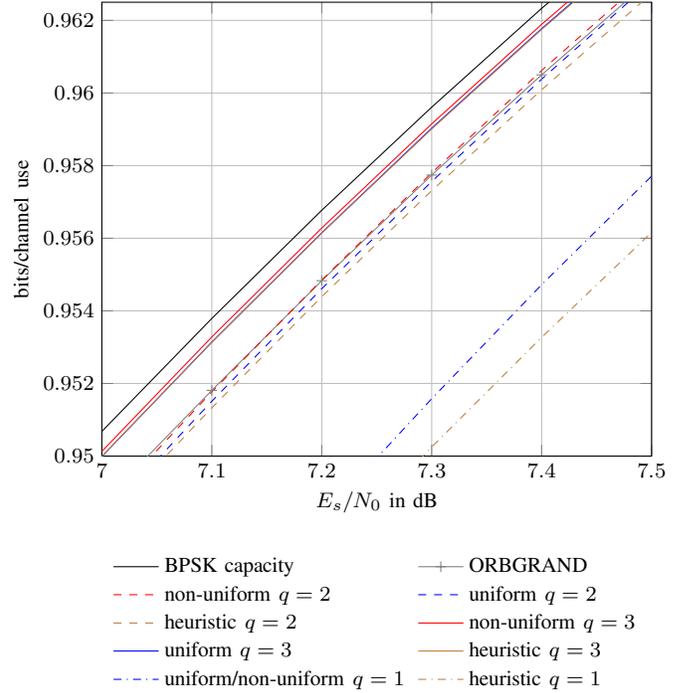

\section{Simulation results}\label{sec:numerical}
For simulated performance evaluation, we consider \ac{BPSK} modulated codewords over real-valued \ac{AWGN} channels.

\subsection{Maximum score of DSGRAND}\label{sec:Smax}
We start with the distribution of the score of the correct noise effect $S(Z^n)$. Since 1) the codewords are uniformly distributed on the codebook $\mathcal{C}$; and 2) the codeword is transmitted over a symmetric channel, i.e., $f_{Y|C}(y|0)=f_{Y|C}(-y|1)$, we could use the all-zero codeword assumption to evaluate the expected score. Assuming that all-zero codeword $0^n$ is transmitted, the \acp{LLR} are \ac{i.i.d.} Gaussian \acp{RV} with mean $2/\sigma^2$ and variance $4/\sigma^2$
\begin{align}
    \tau(Y_i) \overset{\text{i.i.d.}}{\sim} f_{\tau(Y)}(a) = p_\text{G}\left(a\left|\frac{2}{\sigma^2},\frac{4}{\sigma^2}\right.\right).
\end{align}
We have now the score function 
\begin{align}
    S(Z^n)&=\sum_{i=1}^{n} |\tau(Y_i)|\cdot Z_i\\
    &=\sum_{i=1}^{n} \tau(Y_i)\cdot \mathbbm{1}\left\{\tau(Y_i)>0\right\}
\end{align}
where the indicator function $\mathbbm{1}\{\cdot\}$ equals $1$ if the proposition is true and $0$ otherwise. Define mixed-type \acp{RV} $V_i,~i=1,\dots,n$
\begin{align}
    V_i\triangleq \tau(Y_i)\cdot \mathbbm{1}\left\{\tau(Y_i)>0\right\}. 
\end{align}
$V_1,\dots,V_n$ are \ac{i.i.d.}
\begin{equation}
  V_i \overset{\text{i.i.d.}}{\sim}f_{V}(a)=\left\{
    \begin{aligned}
      f_{\tau(Y)}(a), & \text{~if~} a > 0\\
      \delta(a) \int_{-\infty}^0 f_{\tau(Y)}(x)\ \text{d}x, & \text{~if~} a = 0\\
      0, & \text{~otherwise} 
    \end{aligned}\right.
\end{equation}
where $\delta(\cdot)$ denotes the Dirac delta function. Thus, we have the \ac{PDF} of the score $S(Z^n)$
\begin{align}
    f_{S(Z^n)}(a) = f_V(a) ^{\circledast n},
\end{align}
where $(\cdot)^{\circledast n}$ denotes the $n$-th convolution power.

The \ac{GRAND} algorithm is now terminated if the score of the current error pattern exceeds a threshold $s_\text{max}$. The performance loss on \ac{BLER} is upper bounded\footnote{The exact performance loss on \ac{BLER} is the joint probability of event \enquote{ML decoder is capable to find the transmitted codeword for $Y^n$} and \enquote{\ac{GRAND} with a maximum score misdecodes $Y^n$}. Eq.~\eqref{eq:SmaxLoss}, is the probability of event \enquote{\ac{GRAND} with a maximum score misdecodes $Y^n$}, which is a upper bound of the performance loss.} by
\begin{align}\label{eq:SmaxLoss}
   \int_{s_\text{max}}^\infty f_{S(Z^n)}(a)\ \text{d}a.
\end{align}
In addition, such a threshold test enables the decoder to reject a decision when it is not reliable enough, which reduces the number of undetected errors if the threshold is carefully optimized~(see, e.g., \cite{forney1968exponential}). For a given threshold $s_\text{max}$, we define a binary RV $\Phi$ as
\begin{equation}\label{eq:rv_undetected}
    \Phi = \mathbbm{1}\left\{S(z^n)\leq s_\text{max}\right\}.
\end{equation}
The proposition of the indicator function \eqref{eq:rv_undetected} reads as \enquote{GRAND finds an estimate $\hat{c}^n$ with a score smaller than $s_\text{max}$}. Then, the undetected error probability of the algorithm is given as
\begin{align}\label{eq:undetected_error_prob}
    \text{Pr}\left(\hat{C}^n\neq C^n,\Phi = 1\right).
\end{align}
Observe that the overall error probability is equal to the summation of detected and undetected error probabilities, i.e., we have
\begin{align}\label{eq:overall_error_prob}
    \text{Pr}\left(\hat{C}^n\neq U^n\right) = \sum_{\phi\in\{0,1\}}\text{Pr}\left(\hat{C}^n \neq C^n, \Phi = \phi\right)
\end{align}
which simply follows from the law of total probability. The parameter $s_\text{max}$ controls the \ac{BLER} and \ac{uBLER} tradeoff~\cite{forney1968exponential,hof2010performance}. In particular, \eqref{eq:undetected_error_prob} becomes equal to the left-hand side of~\eqref{eq:overall_error_prob} if $s_\text{max} = \infty$.

For a \ac{DSGRAND}, the \acp{LLR} are quantized by function $\bar{\mu}(\cdot)$. The \ac{PMF} of discrete \ac{RV} $\bar{\mu}(V)$ is given by
\begin{equation}f_{\bar{\mu}(V)}(a)=\left\{
    \begin{aligned}
        \int_{-\infty}^0 f_{\tau(Y)}(x)\ \text{d}x, &\text{~if~}a = 0\\
        \int_{b_{i-1}}^{b_i}f_{\tau(Y)}(x)\ \text{d}x, &\text{~if~}a = v_1,\dots,v_Q
    \end{aligned}\right.
\end{equation}
where $Q,b_i,v_i$ are the quantization parameters (see Section~\ref{sec:quantizer}). Thus, we have the \ac{PMF} of the score $S(Z^n)$
\begin{align}
    f_{S(Z^n)}(a) = f_{\bar{\mu}(V)}(a) ^{\circledast n}.
\end{align}
In our simulations of \ac{DSGRAND}, we always use a $s_\text{max}$ such that the performance loss eq.~\eqref{eq:SmaxLoss} is close to the target \ac{BLER}.

\subsection{Error correction performance}
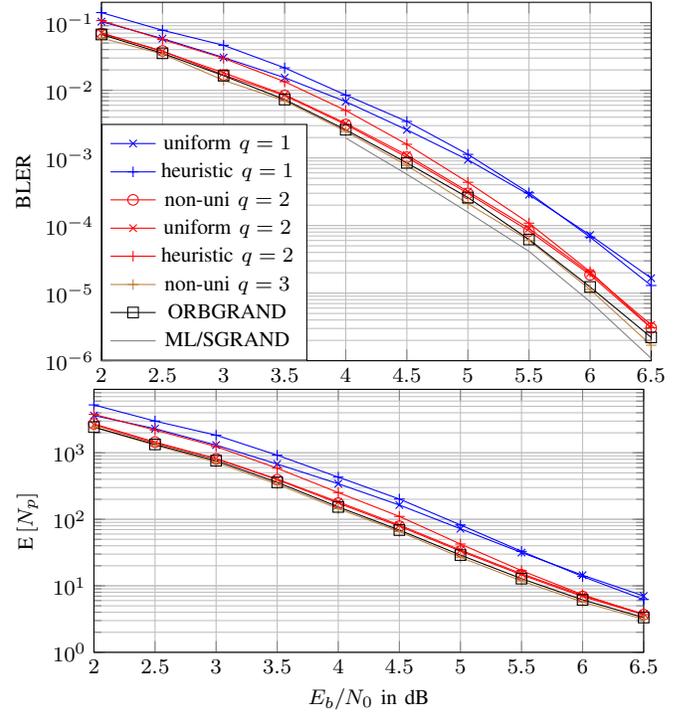
\begin{figure}[t]
	\centering
	\footnotesize
	\begin{tikzpicture}[scale=1]
\begin{semilogyaxis}[
legend style={at={(0,0)},anchor= south west},
ymin=0.000001,
ymax=0.2,
width=3.5in,
height=2.5in,
grid=both,
xmin = 2.0,
xmax = 6.5,
ylabel = BLER,
xtick=data,
]

\addplot[blue, mark = x]
table[x=snr,y=fer]{snr fer
2 0.10352
2.5 0.058166
3 0.030529
3.5 0.015356
4 0.0067853
4.5 0.0025976
5 0.00093596
5.5 0.00028534
6 7.2221e-05
6.5 1.6643e-05
};\addlegendentry{uniform $q=1$}

\addplot[blue, mark = +]
table[x=snr,y=fer]{snr fer
2 0.14079
2.5 0.077267
3 0.046249
3.5 0.021463
4 0.0085062
4.5 0.0034609
5 0.0011267
5.5 0.00030395
6 6.741e-05
6.5 1.2955e-05
};\addlegendentry{heuristic $q=1$}

\addplot[red, mark = o]
table[x=snr,y=fer]{snr fer
2 0.06921
2.5 0.038203
3 0.016843
3.5 0.0083758
4 0.0031099
4.5 0.001022
5 0.00029829
5.5 8.2813e-05
6 1.8678e-05
6.5 3.0522e-06
};\addlegendentry{non-uni $q=2$}

\addplot[red, mark = x]
table[x=snr,y=fer]{snr fer
2 0.070689
2.5 0.037515
3 0.018177
3.5 0.0085852
4 0.0032452
4.5 0.0010994
5 0.00031831
5.5 9.0314e-05
6 1.9753e-05
6.5 3.3828e-06
};\addlegendentry{uniform $q=2$}

\addplot[red, mark = +]
table[x=snr,y=fer]{snr fer
2 0.10737
2.5 0.056273
3 0.029757
3.5 0.013401
4 0.0050275
4.5 0.0015872
5 0.00043531
5.5 0.00010832
6 2.0715e-05
6.5 3.0776e-06
};\addlegendentry{heuristic $q=2$}

\addplot[brown, mark = +]
table[x=snr,y=fer]{snr fer
2 0.059154
2.5 0.034417
3 0.014036
3.5 0.0069798
4 0.0024585
4.5 0.00077424
5 0.0002108
5.5 6.001e-05
6 1.132e-05
6.5 1.6957e-06
};\addlegendentry{non-uni $q=3$}

\addplot[black, mark = square]
table[x=snr,y=fer]{snr fer
2 0.066844
2.5 0.035278
3 0.016422
3.5 0.0072939
4 0.002606
4.5 0.00084392
5 0.00025929
5.5 6.211e-05
6 1.2339e-05
6.5 2.2128e-06
};\addlegendentry{ORBGRAND}

\addplot[gray]
table[x=snr,y=fer]{snr fer
4 0.0019668
4.5 0.00057294
5 0.00015705
5.5 4.1407e-05
6 7.4712e-06
6.5 1.0852e-06
};\addlegendentry{ML/SGRAND}

\end{semilogyaxis}
\end{tikzpicture}

\begin{tikzpicture}[scale=1]
\begin{semilogyaxis}[
ymin=1,
ymax=9000,
width=3.5in,
height=2in,
grid=both,
xmin = 2,
xmax = 6.5,
xlabel = $E_b/N_0$ in dB,
ylabel = {$\text{E}\left[N_p\right]$},
xtick=data,
]

\addplot[blue, mark = x]
table[x=snr,y=fer]{snr fer
2 3633.6
2.5 2309.2
3 1306.6
3.5 677.62
4 342.69
4.5 164.04
5 72.031
5.5 31.624
6 14.45
6.5 7.04
};

\addplot[blue, mark = +]
table[x=snr,y=fer]{snr fer
2 5251.3
2.5 3000.5
3 1841.2
3.5 921.3
4 430.71
4.5 201.76
5 82.322
5.5 33.013
6 13.79
6.5 6.2201
};

\addplot[red, mark = o]
table[x=snr,y=fer]{snr fer
2 2720.7
2.5 1451.3
3 812.29
3.5 394.83
4 174.91
4.5 78.62
5 33.231
5.5 14.629
6 6.8953
6.5 3.6809
};

\addplot[red, mark = x]
table[x=snr,y=fer]{snr fer
2 2639.2
2.5 1406.7
3 819.6
3.5 401.94
4 182.68
4.5 81.91
5 34.475
5.5 15.205
6 7.1488
6.5 3.7784
};

\addplot[red, mark = +]
table[x=snr,y=fer]{snr fer
2 3780.4
2.5 2204.1
3 1243.4
3.5 586.06
4 252.18
4.5 110.52
5 42.475
5.5 16.926
6 7.3653
6.5 3.6848
};

\addplot[brown, mark = +]
table[x=snr,y=fer]{snr fer
2 2415.8
2.5 1300.9
3 719.97
3.5 336.68
4 144.1
4.5 65.339
5 26.781
5.5 11.832
6 5.7024
6.5 3.1383
};

\addplot[black, mark = square]
table[x=snr,y=fer]{snr fer
2 2437.1
2.5 1338.1
3 761.87
3.5 359.8
4 153.43
4.5 68.891
5 28.992
5.5 12.807
6 6.1544
6.5 3.3331
};

\end{semilogyaxis}
\end{tikzpicture}
	\caption{$(31,16)$ BCH code under DSGRANDs and ORBGRAND.}
	\label{fig:bch_31_16}
\end{figure}

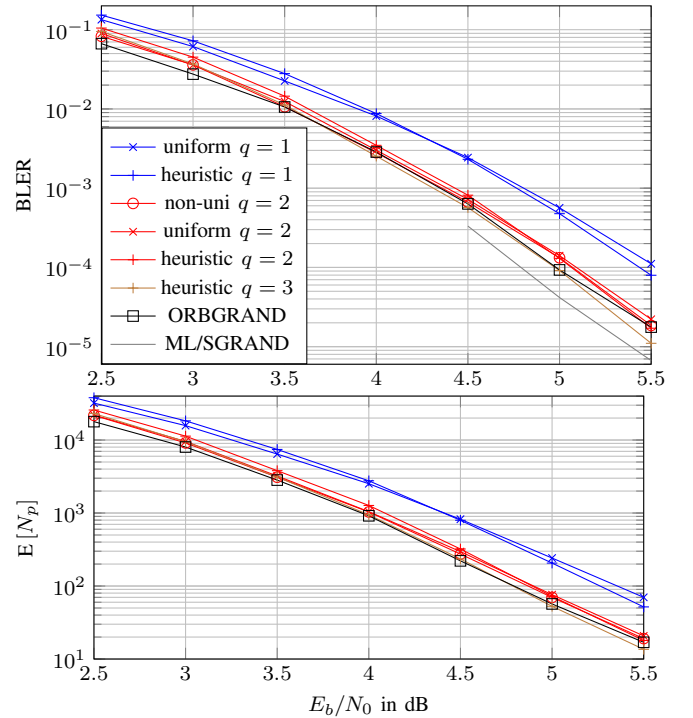
\begin{figure}[t]
	\centering
	\footnotesize
	\begin{tikzpicture}[scale=1]
\begin{semilogyaxis}[
legend style={at={(0,0)},anchor= south west},
ymin=0.000006,
ymax=0.2,
width=3.5in,
height=2.5in,
grid=both,
xmin = 2.5,
xmax = 5.5,
ylabel = BLER
]

\addplot[blue, mark = x]
table[x=snr,y=fer]{snr fer
2.5 0.13384
3 0.0613
3.5 0.022583
4 0.0081381
4.5 0.0024322
5 0.00056486
5.5 0.00011115
};\addlegendentry{uniform $q=1$}

\addplot[blue, mark = +]
table[x=snr,y=fer]{snr fer
2.5 0.15282
3 0.072378
3.5 0.02792
4 0.0087466
4.5 0.0023156
5 0.00047673
5.5 7.9668e-05
};\addlegendentry{heuristic $q=1$}

\addplot[red, mark = o]
table[x=snr,y=fer]{snr fer
2.5 0.082582
3 0.036004
3.5 0.010674
4 0.0028014
4.5 0.0006685
5 0.00013242
5.5 1.8525e-05
};\addlegendentry{non-uni $q=2$}

\addplot[red, mark = x]
table[x=snr,y=fer]{snr fer
2.5 0.089701
3 0.03582
3.5 0.012415
4 0.0030169
4.5 0.00072255
5 0.00014064
5.5 2.2164e-05
};\addlegendentry{uniform $q=2$}

\addplot[red, mark = +]
table[x=snr,y=fer]{snr fer
2.5 0.10489
3 0.045052
3.5 0.014494
4 0.0034733
4.5 0.00081642
5 0.00013014
5.5 1.7091e-05

};\addlegendentry{heuristic $q=2$}

\addplot[brown, mark = +]
table[x=snr,y=fer]{snr fer
2.5 0.094922
3 0.036928
3.5 0.011235
4 0.0025352
4.5 0.00056894
5 9.1327e-05
5.5 1.1027e-05
};\addlegendentry{heuristic $q=3$}

\addplot[black, mark = square]
table[x=snr,y=fer]{snr fer
2.5 0.066643
3 0.02753
3.5 0.010586
4 0.0028599
4.5 0.00063354
5 9.3016e-05
5.5 1.7666e-05
};\addlegendentry{ORBGRAND}

\addplot[gray]
table[x=snr,y=fer]{snr fer
4.5 0.00033228
5 4.1697e-05
5.5 6.6785e-06
};\addlegendentry{ML/SGRAND}

\end{semilogyaxis}
\end{tikzpicture}

\begin{tikzpicture}[scale=1]
\begin{semilogyaxis}[
ymin=10,
ymax=40000,
width=3.5in,
height=2in,
grid=both,
xmin = 2.5,
xmax = 5.5,
xlabel = $E_b/N_0$ in dB,
ylabel = {$\text{E}\left[N_p\right]$},
]

\addplot[blue, mark = x]
table[x=snr,y=fer]{snr fer
2.5 32084
3 15755
3.5 6421.3
4 2521.1
4.5 822.25
5 241.82
5.5 70.03
};

\addplot[blue, mark = +]
table[x=snr,y=fer]{snr fer
2.5 37996
3 18318
3.5 7430.9
4 2762.9
4.5 795.71
5 206.25
5.5 51.766
};

\addplot[red, mark = o]
table[x=snr,y=fer]{snr fer
2.5 21764
3 9126.4
3.5 3095.4
4 1034
4.5 274.12
5 69.901
5.5 19.063
};

\addplot[red, mark = x]
table[x=snr,y=fer]{snr fer
2.5 21997
3 9564.8
3.5 3217.4
4 1044
4.5 296.48
5 76.16
5.5 20.743
};

\addplot[red, mark = +]
table[x=snr,y=fer]{snr fer
2.5 25792
3 11310
3.5 3806.1
4 1269.8
4.5 321.58
5 72.374
5.5 17.966
};

\addplot[brown, mark = +]
table[x=snr,y=fer]{snr fer
2.5 23079
3 9476
3.5 3116.9
4 945.42
4.5 235.12
5 52.745
5.5 13.585
};

\addplot[black, mark = square]
table[x=snr,y=fer]{snr fer
2.5 17873
3 8018
3.5 2827
4 916.33
4.5 221.22
5 56.822
5.5 16.958

};

\end{semilogyaxis}
\end{tikzpicture}
	\caption{$(63,45)$ BCH code under DSGRANDs and ORBGRAND.}
	\label{fig:bch_63_45}
\end{figure}

\begin{figure}[t]
	\centering
	\footnotesize
	\begin{tikzpicture}[scale=1]
\begin{semilogyaxis}[
legend style={at={(0,0)},anchor= south west},
ymin=0.000002,
ymax=0.5,
width=3.5in,
height=2.5in,
grid=both,
xmin = 3.0,
xmax = 6.5,
ylabel = BLER
]

\addplot[blue, mark = x]
table[x=snr,y=fer]{snr fer
3 0.44462
3.5 0.23723
4 0.11143
4.5 0.041773
5 0.011888
5.5 0.0030357
6 0.00063312
6.5 0.00011462
};\addlegendentry{uniform $q=1$}

\addplot[blue, mark = +]
table[x=snr,y=fer]{snr fer
3 0.43323
3.5 0.24015
4 0.11079
4.5 0.040646
5 0.010375
5.5 0.0023181
6 0.00039807
6.5 4.8858e-05
};\addlegendentry{heuristic $q=1$}

\addplot[red, mark = o]
table[x=snr,y=fer]{snr fer
3 0.32736
3.5 0.15474
4 0.070649
4.5 0.017581
5 0.0038502
5.5 0.00071549
6 0.00011156
6.5 1.314e-05
};\addlegendentry{non-uni $q=2$}

\addplot[red, mark = x]
table[x=snr,y=fer]{snr fer
3 0.33876
3.5 0.16423
4 0.07097
4.5 0.020436
5 0.0041474
5.5 0.00082509
6 0.00013168
6.5 1.6903e-05
};\addlegendentry{uniform $q=2$}

\addplot[red, mark = +]
table[x=snr,y=fer]{snr fer
3 0.34853
3.5 0.17226
4 0.069364
4.5 0.018708
5 0.003499
5.5 0.00060176
6 7.5545e-05
6.5 6.8115e-06
};\addlegendentry{heuristic $q=2$}

\addplot[brown, mark = +]
table[x=snr,y=fer]{snr fer
3 0.32573
3.5 0.14599
4 0.064226
4.5 0.015026
5 0.0027019
5.5 0.00041358
6 4.6777e-05
6.5 4.2046e-06
};\addlegendentry{heuristic $q=3$}

\addplot[black, mark = square]
table[x=snr,y=fer]{snr fer
3 0.33938
3.5 0.16196
4 0.058055
4.5 0.017662
5 0.0046943
5.5 0.00099729
6 0.0001845
6.5 3.3795e-05

};\addlegendentry{ORBGRAND}

\addplot[gray]
table[x=snr,y=fer]{snr fer
4.5 0.010429
5 0.0018977
5.5 0.00027871
6 2.9689e-05
6.5 2.2574e-06

};\addlegendentry{ML/SGRAND}

\end{semilogyaxis}
\end{tikzpicture}

\begin{tikzpicture}[scale=1]
\begin{semilogyaxis}[
ymin=1,
ymax=9000,
width=3.5in,
height=2in,
grid=both,
xmin = 3.0,
xmax = 6.5,
xlabel = $E_b/N_0$ in dB,
ylabel = {$\text{E}\left[N_p\right]$}
]

\addplot[blue, mark = x]
table[x=snr,y=fer]{snr fer
3 6538.1
3.5 3824.6
4 1742.1
4.5 803.99
5 262.52
5.5 81.09
6 26.101
6.5 9.0639
};

\addplot[blue, mark = +]
table[x=snr,y=fer]{snr fer
3 6333.3
3.5 3720.9
4 1876.1
4.5 761
5 224.46
5.5 60.008
6 15.942
6.5 4.8895
};

\addplot[red, mark = o]
table[x=snr,y=fer]{snr fer
3 5508.4
3.5 2816.1
4 1134.5
4.5 360.3
5 90.155
5.5 23.199
6 6.9337
6.5 2.7114
};

\addplot[red, mark = x]
table[x=snr,y=fer]{snr fer
3 6170.2
3.5 2999.9
4 1173.6
4.5 384.64
5 100.32
5.5 25.79
6 7.7435
6.5 3.0309
};

\addplot[red, mark = +]
table[x=snr,y=fer]{snr fer
3 5293.7
3.5 2946.8
4 1118.6
4.5 353.77
5 85.009
5.5 20.08
6 5.4471
6.5 2.2587
};

\addplot[brown, mark = +]
table[x=snr,y=fer]{snr fer
3 5317.2
3.5 2472.6
4 1013.5
4.5 284.74
5 64.021
5.5 14.896
6 4.3009
6.5 1.9736
};

\addplot[black, mark = square]
table[x=snr,y=fer]{snr fer
3 5469.7
3.5 2627.7
4 1015.1
4.5 377.94
5 104.81
5.5 25.742
6 7.2224
6.5 2.6382

};

\end{semilogyaxis}
\end{tikzpicture}
	\caption{$(127,113)$ BCH code under DSGRANDs and ORBGRAND.}
	\label{fig:bch_127_113}
\end{figure}

\begin{figure}[t]
	\centering
	\footnotesize
	\begin{tikzpicture}[scale=1]
\begin{semilogyaxis}[
legend style={at={(0,0)},anchor= south west},
ymin=0.000005,
ymax=0.11,
width=3.5in,
height=2.5in,
grid=both,
xmin = 3.5,
xmax = 5.5,
ylabel = BLER
]

\addplot[blue, mark = x]
table[x=snr,y=fer]{snr fer
3.5 0.10165
4 0.033225
4.5 0.0079965
5 0.0016051
5.5 0.00024525
};\addlegendentry{uniform $q=1$}

\addplot[blue, mark = +]
table[x=snr,y=fer]{snr fer
3.5 0.10165
4 0.030836
4.5 0.0074337
5 0.0010995
5.5 0.00013437
};\addlegendentry{heuristic $q=1$}

\addplot[red, mark = o]
table[x=snr,y=fer]{snr fer
3.5 0.050413
4 0.012161
4.5 0.0021324
5 0.00023807
5.5 2.767e-05
};\addlegendentry{non-uni $q=2$}

\addplot[red, mark = x]
table[x=snr,y=fer]{snr fer
3.5 0.045455
4 0.012595
4.5 0.0024285
5 0.00029942
5.5 3.5036e-05
};\addlegendentry{uniform $q=2$}

\addplot[red, mark = +]
table[x=snr,y=fer]{snr fer
3.5 0.047107
4 0.012595
4.5 0.0022805
5 0.00020125
5.5 1.7717e-05
};\addlegendentry{heuristic $q=2$}

\addplot[brown, mark = +]
table[x=snr,y=fer]{snr fer
3.5 0.041322
4 0.010858
4.5 0.0014808
5 0.00012271
5.5 9.9534e-06
};\addlegendentry{heuristic $q=3$}

\addplot[black, mark = square]
table[x=snr,y=fer]{snr fer
3.5 0.041322
4 0.011943
4.5 0.0021028
5 0.00033624
5.5 5.6934e-05
};\addlegendentry{ORBGRAND}

\addplot[gray]
table[x=snr,y=fer]{snr fer
5 6.8721e-05
5.5 5.972e-06
};\addlegendentry{ML/SGRAND}

\end{semilogyaxis}
\end{tikzpicture}

\begin{tikzpicture}[scale=1]
\begin{semilogyaxis}[
ymin=50,
ymax=110000,
width=3.5in,
height=2in,
grid=both,
xmin = 3.5,
xmax = 5.5,
xlabel = $E_b/N_0$ in dB,
ylabel = {$\text{E}\left[N_p\right]$}
]

\addplot[blue, mark = x]
table[x=snr,y=fer]{snr fer
3.5 1.0087e+05
4 36475
4.5 11321
5 2714.5
5.5 579.87
};

\addplot[blue, mark = +]
table[x=snr,y=fer]{snr fer
3.5 1.0343e+05
4 35566
4.5 10252
5 2061.4
5.5 353.01
};

\addplot[red, mark = o]
table[x=snr,y=fer]{snr fer
3.5 59586
4 16647
4.5 3433.1
5 595.08
5.5 98.529
};

\addplot[red, mark = x]
table[x=snr,y=fer]{snr fer
3.5 60183
4 16480
4.5 3791.4
5 669.32
5.5 115
};

\addplot[red, mark = +]
table[x=snr,y=fer]{snr fer
3.5 60591
4 15931
4.5 3479.3
5 539.61
5.5 75.639
};

\addplot[brown, mark = +]
table[x=snr,y=fer]{snr fer
3.5 52161
4 13587
4.5 2436.3
5 348.2
5.5 50.49
};

\addplot[black, mark = square]
table[x=snr,y=fer]{snr fer
3.5 54182
4 14535
4.5 3480.7
5 673.44
5.5 134.19
};

\end{semilogyaxis}
\end{tikzpicture}
	\caption{$(127,106)$ BCH code under DSGRANDs and ORBGRAND.}
	\label{fig:bch_127_106}
\end{figure}

\begin{figure}[t]
	\centering
	\footnotesize
	\begin{tikzpicture}[scale=1]
\begin{semilogyaxis}[
legend style={at={(0,0)},anchor= south west},
ymin=0.0000005,
ymax=0.4,
width=3.5in,
height=2.5in,
grid=both,
xmin = 4.0,
xmax = 7.0,
ylabel = BLER
]

\addplot[blue, mark = x]
table[x=snr,y=fer]{snr fer
4 0.37979
4.5 0.17389
5 0.056227
5.5 0.014016
6 0.0027405
6.5 0.00046071
7 6.743e-05
};\addlegendentry{uniform $q=1$}

\addplot[blue, mark = +]
table[x=snr,y=fer]{snr fer
4 0.38619
4.5 0.18216
5 0.058059
5.5 0.013786
6 0.0021066
6.5 0.00025981
7 2.242e-05
};\addlegendentry{heuristic $q=1$}

\addplot[red, mark = o]
table[x=snr,y=fer]{snr fer
4 0.27749
4.5 0.090045
5 0.023718
5.5 0.0040784
6 0.00049354
6.5 5.3612e-05
7 4.68e-06
};\addlegendentry{non-uni $q=2$}

\addplot[red, mark = x]
table[x=snr,y=fer]{snr fer
4 0.28772
4.5 0.097893
5 0.026374
5.5 0.0042852
6 0.00060335
6.5 7.4288e-05
7 6.99e-06
};\addlegendentry{uniform $q=2$}

\addplot[red, mark = +]
table[x=snr,y=fer]{snr fer
4 0.29156
4.5 0.096654
5 0.023535
5.5 0.003527
6 0.00036902
6.5 2.917e-05
7 1.805e-06
};\addlegendentry{heuristic $q=2$}

\addplot[brown, mark = +]
table[x=snr,y=fer]{snr fer
4 0.25575
4.5 0.082611
5 0.018315
5.5 0.0022977
6 0.0002264
6.5 1.6028e-05
7 8.8e-07
};\addlegendentry{heuristic $q=3$}

\addplot[black, mark = square]
table[x=snr,y=fer]{snr fer
4 0.28508
4.5 0.10039
5 0.024543
5.5 0.0050338
6 0.0009855
6.5 0.00015496
7 1.8465e-05
};\addlegendentry{ORBGRAND}

\addplot[gray]
table[x=snr,y=fer]{snr fer
5.5 0.0016586
6 0.00011826
6.5 1.1247e-05
7 5.3465e-07
};\addlegendentry{ML/SGRAND}

\end{semilogyaxis}
\end{tikzpicture}

\begin{tikzpicture}[scale=1]
\begin{semilogyaxis}[
ytick={1,10,100,1000,10000,100000},
ymin=1,
ymax=30000,
width=3.5in,
height=2in,
grid=both,
xmin = 4.0,
xmax = 7.0,
xlabel = $E_b/N_0$ in dB,
ylabel = {$\text{E}\left[N_p\right]$}
]

\addplot[blue, mark = x]
table[x=snr,y=fer]{snr fer
4 24441
4.5 10626
5 3676.3
5.5 1073.8
6 256.38
6.5 62.056
7 16.665
};

\addplot[blue, mark = +]
table[x=snr,y=fer]{snr fer
4 25162
4.5 12053
5 3878.7
5.5 1009.9
6 181.72
6.5 31.898
7 6.5302
};

\addplot[red, mark = o]
table[x=snr,y=fer]{snr fer
4 18422
4.5 6251.6
5 1609.3
5.5 329.22
6 52.114
6.5 10.333
7 3.059
};

\addplot[red, mark = x]
table[x=snr,y=fer]{snr fer
4 17807
4.5 6431.8
5 1624.3
5.5 361.55
6 60.16
6.5 12.622
7 3.7267
};

\addplot[red, mark = +]
table[x=snr,y=fer]{snr fer
4 19511
4.5 6464.2
5 1572.2
5.5 278.73
6 38.264
6.5 6.8142
7 2.2338
};

\addplot[brown, mark = +]
table[x=snr,y=fer]{snr fer
4 17334
4.5 5432.4
5 1221
5.5 202.34
6 25.608
6.5 4.8459
7 1.8791
};

\addplot[black, mark = square]
table[x=snr,y=fer]{snr fer
4 16638
4.5 6486.8
5 1738.3
5.5 353.38
6 77.822
6.5 15.715
7 3.8631
};

\end{semilogyaxis}
\end{tikzpicture}
	\caption{$(255,239)$ BCH code under DSGRANDs and ORBGRAND.}
	\label{fig:bch_255_239}
\end{figure}
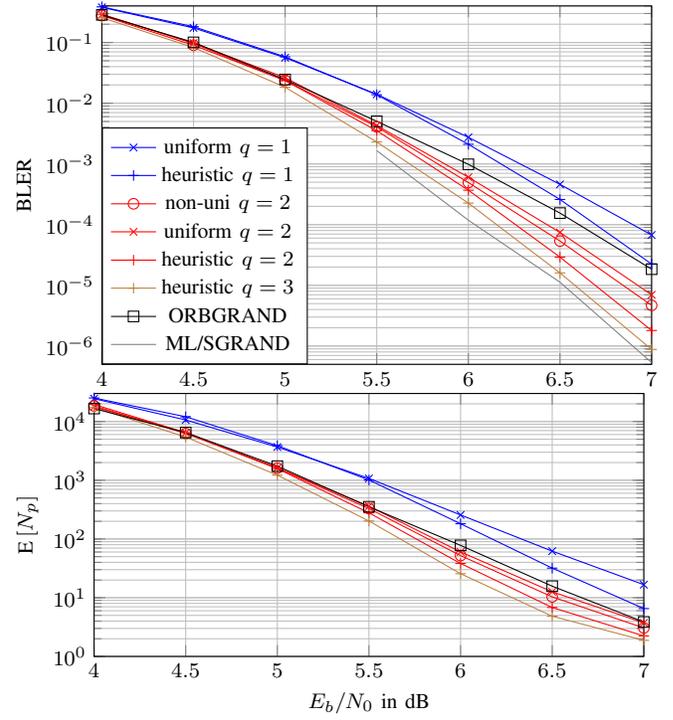

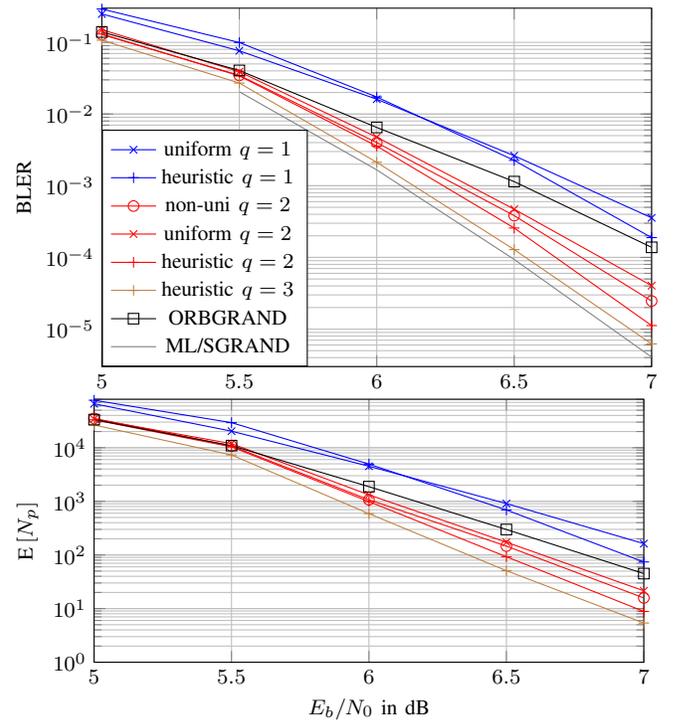
\begin{figure}[t]
	\centering
	\footnotesize
	\begin{tikzpicture}[scale=1]
\begin{semilogyaxis}[
legend style={at={(0,0)},anchor= south west},
ymin=0.000003,
ymax=0.3,
width=3.5in,
height=2.5in,
grid=both,
xmin = 5.0,
xmax = 7.0,
ylabel = BLER
]

\addplot[blue, mark = x]
table[x=snr,y=fer]{snr fer
5 0.24893
5.5 0.076299
6 0.016109
6.5 0.0026218
7 0.00035688
};\addlegendentry{uniform $q=1$}

\addplot[blue, mark = +]
table[x=snr,y=fer]{snr fer
5 0.29185
5.5 0.099567
6 0.017299
6.5 0.0022594
7 0.00018896
};\addlegendentry{heuristic $q=1$}

\addplot[red, mark = o]
table[x=snr,y=fer]{snr fer
5 0.12876
5.5 0.034632
6 0.0039952
6.5 0.00038299
7 2.4647e-05
};\addlegendentry{non-uni $q=2$}

\addplot[red, mark = x]
table[x=snr,y=fer]{snr fer
5 0.15021
5.5 0.038961
6 0.0047603
6.5 0.00047039
7 4.0331e-05
};\addlegendentry{uniform $q=2$}

\addplot[red, mark = +]
table[x=snr,y=fer]{snr fer
5 0.1309
5.5 0.034091
6 0.0035702
6.5 0.00025704
7 1.1203e-05
};\addlegendentry{heuristic $q=2$}

\addplot[brown, mark = +]
table[x=snr,y=fer]{snr fer
5 0.1073
5.5 0.027056
6 0.0021251
6.5 0.00012852
7 6.2239e-06
};\addlegendentry{heuristic $q=3$}

\addplot[black, mark = square]
table[x=snr,y=fer]{snr fer
5 0.13948
5.5 0.040584
6 0.0065029
6.5 0.0011413
7 0.0001383
};\addlegendentry{ORBGRAND}

\addplot[gray]
table[x=snr,y=fer]{snr fer
5.5 2.048e-02
6 1.664e-03
6.5 9.366e-05
7 4.063e-06
};\addlegendentry{ML/SGRAND}

\end{semilogyaxis}
\end{tikzpicture}

\begin{tikzpicture}[scale=1]
\begin{semilogyaxis}[
ytick={1,10,100,1000,10000,100000},
ymin=1,
ymax=80000,
width=3.5in,
height=2in,
grid=both,
xmin = 5.0,
xmax = 7.0,
xlabel = $E_b/N_0$ in dB,
ylabel = {$\text{E}\left[N_p\right]$}
]

\addplot[blue, mark = x]
table[x=snr,y=fer]{snr fer
5 66018
5.5 20324
6 4510.2
6.5 908.44
7 162.61
};

\addplot[blue, mark = +]
table[x=snr,y=fer]{snr fer
5 76780
5.5 29303
6 4945
6.5 689.95
7 74.043
};

\addplot[red, mark = o]
table[x=snr,y=fer]{snr fer
5 35202
5.5 10756
6 1067.9
6.5 146.15
7 15.806
};

\addplot[red, mark = x]
table[x=snr,y=fer]{snr fer
5 34240
5.5 11903
6 1320.9
6.5 172.65
7 21.433
};

\addplot[red, mark = +]
table[x=snr,y=fer]{snr fer
5 33136
5.5 10364
6 987.65
6.5 92.497
7 8.8528
};

\addplot[brown, mark = +]
table[x=snr,y=fer]{snr fer
5 26509
5.5 7310.3
6 588.78
6.5 50.859
7 5.3773
};

\addplot[black, mark = square]
table[x=snr,y=fer]{snr fer
5 33372
5.5 10922
6 1877.4
6.5 298.1
7 44.93
};

\end{semilogyaxis}
\end{tikzpicture}
	\caption{$(511,493)$ BCH code under DSGRANDs and ORBGRAND.}
	\label{fig:bch_511}
\end{figure}

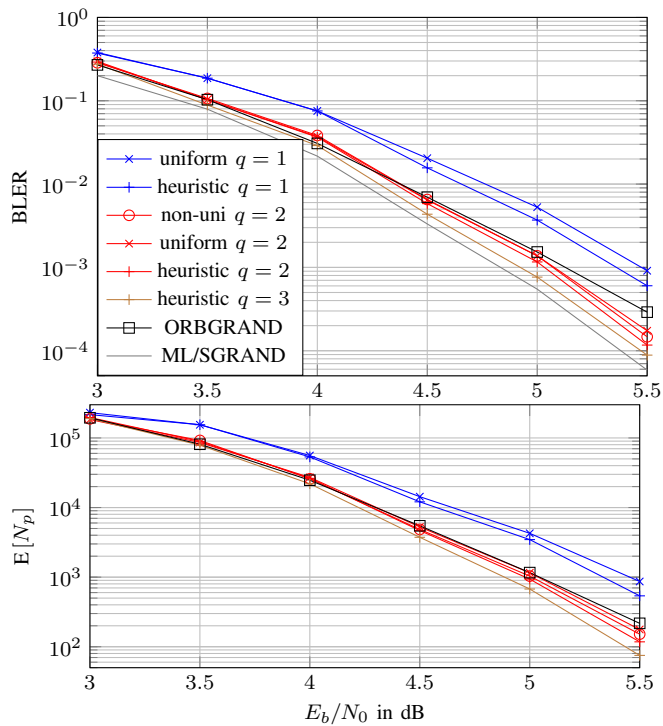
\begin{figure}[t]
	\centering
	\footnotesize
	\begin{tikzpicture}[scale=1]
\begin{semilogyaxis}[
legend style={at={(0,0)},anchor= south west},
ymin=0.00005,
ymax=1,
width=3.5in,
height=2.5in,
grid=both,
xmin = 3,
xmax = 5.5,
ylabel = BLER
]

\addplot[blue, mark = x]
table[x=snr,y=fer]{snr fer
3 0.37912
3.5 0.1865
4 0.07545
4.5 0.020445
5 0.0052794
5.5 0.00091259
};\addlegendentry{uniform $q=1$}

\addplot[blue, mark = +]
table[x=snr,y=fer]{snr fer
3 0.37363
3.5 0.1865
4 0.07545
4.5 0.01566
5 0.0036979
5.5 0.00060248
};\addlegendentry{heuristic $q=1$}

\addplot[red, mark = o]
table[x=snr,y=fer]{snr fer
3 0.28571
3.5 0.10657
4 0.038305
4.5 0.0065251
5 0.0013772
5.5 0.00014708
};\addlegendentry{non-uni $q=2$}

\addplot[red, mark = x]
table[x=snr,y=fer]{snr fer
3 0.2967
3.5 0.10302
4 0.037144
4.5 0.0064381
5 0.0013772
5.5 0.00017366
};\addlegendentry{uniform $q=2$}

\addplot[red, mark = +]
table[x=snr,y=fer]{snr fer
3 0.29121
3.5 0.10302
4 0.036564
4.5 0.0058291
5 0.001163
5.5 0.00011695
};\addlegendentry{heuristic $q=2$}

\addplot[brown, mark = +]
table[x=snr,y=fer]{snr fer
3 0.27472
3.5 0.08881
4 0.029019
4.5 0.0043501
5 0.00076513
5.5 8.8601e-05
};\addlegendentry{heuristic $q=3$}

\addplot[black, mark = square]
table[x=snr,y=fer]{snr fer
3 0.26923
3.5 0.10302
4 0.03076
4.5 0.0069602
5 0.0015303
5.5 0.00029061
};\addlegendentry{ORBGRAND}

\addplot[gray]
table[x=snr,y=fer]{snr fer
3 2.000e-01
3.5 7.924e-02
4 2.165e-02
4.5 3.324e-03
5 5.485e-04
5.5 5.8623e-05
};\addlegendentry{ML/SGRAND}

\end{semilogyaxis}
\end{tikzpicture}

\begin{tikzpicture}[scale=1]
\begin{semilogyaxis}[
ymin=50,
ymax=3e+05,
width=3.5in,
height=2in,
grid=both,
xmin = 3,
xmax = 5.5,
xlabel = $E_b/N_0$ in dB,
ylabel = {$\text{E}\left[N_p\right]$}
]

\addplot[blue, mark = x]
table[x=snr,y=fer]{snr fer
3 2.3106e+05
3.5 1.547e+05
4 55872
4.5 14332
5 4263.5
5.5 866.15
};

\addplot[blue, mark = +]
table[x=snr,y=fer]{snr fer
3 2.1636e+05
3.5 1.5495e+05
4 53395
4.5 12079
5 3464.8
5.5 538.73
};

\addplot[red, mark = o]
table[x=snr,y=fer]{snr fer
3 1.8786e+05
3.5 92127
4 25889
4.5 4851.2
5 1035.2
5.5 151.76
};

\addplot[red, mark = x]
table[x=snr,y=fer]{snr fer
3 1.8395e+05
3.5 86085
4 27071
4.5 5257.8
5 1155.5
5.5 177.16
};

\addplot[red, mark = +]
table[x=snr,y=fer]{snr fer
3 1.9674e+05
3.5 89177
4 26290
4.5 4674.8
5 937.33
5.5 117.61
};

\addplot[brown, mark = +]
table[x=snr,y=fer]{snr fer
3 1.8975e+05
3.5 78826
4 21943
4.5 3736.4
5 675.22
5.5 75.099
};

\addplot[black, mark = square]
table[x=snr,y=fer]{snr fer
3 1.9492e+05
3.5 81253
4 24696
4.5 5465.6
5 1161.8
5.5 216.93
};

\end{semilogyaxis}
\end{tikzpicture}
	\caption{$(31,27)$ RS code on $\text{GF}\left(2^5\right)$ under DSGRANDs and ORBGRAND.}
	\label{fig:rs_155}
\end{figure}

\begin{figure}[t]
	\centering
	\footnotesize
	\begin{tikzpicture}[scale=1]
\begin{semilogyaxis}[
legend style={at={(0.5,-0.2)},anchor=north,draw=none,/tikz/every even column/.append style={column sep=1mm},cells={align=left}},
legend columns=2,
legend cell align=left,
ymin=0.00001,
ymax=0.3,
width=3.5in,
height=3.0in,
grid=both,
xmin = 3.5,
xmax = 5.5,
xlabel = $E_b/N_0$ in dB,
ylabel = BLER
]


\addplot[blue, mark = +]
table[x=snr,y=fer]{snr fer
3.5 0.11642
4 0.034375
4.5 0.0079216
5 0.0015811
5.5 0.00025685
};\addlegendentry{heuristic $q=1$}



\addplot[blue, mark = x]
table[x=snr,y=fer]{snr fer
3.5 0.057173
4 0.014375
4.5 0.0025085
5 0.00036611
5.5 4.5604e-05
};\addlegendentry{heuristic $q=2$}

\addplot[blue, mark = o]
table[x=snr,y=fer]{snr fer
3.5 0.051975
4 0.010417
4.5 0.0016503
5 0.00020116
5.5 2.6209e-05
};\addlegendentry{heuristic $q=3$}

\addplot[black, mark = square]
table[x=snr,y=fer]{snr fer
3.5 0.054054
4 0.01375
4.5 0.0022444
5 0.0004506
5.5 7.8103e-05
};\addlegendentry{ORBGRAND}

\addplot[gray]
table[x=snr,y=fer]{snr fer
5 0.00015058
5.5 1.925e-05
};\addlegendentry{ML/SGRAND}

\addplot[red, dashed, mark options=solid,mark = +]
table[x=snr,y=fer]{snr fer
3.5 0.250626566416040
4 0.0943396226415094
4.5 0.0315258511979823
5 0.00647039792947266
5.5 0.000907754034966686
};\addlegendentry{SCL, $L=8$}

\addplot[red, dashed, mark options=solid,mark = x]
table[x=snr,y=fer]{snr fer
3.5 0.130718954248366
4 0.0452284034373587
4.5 0.00898109479545557
5 0.00137596059248863
5.5 0.000185000000000000
};\addlegendentry{SCL, $L=32$}

\addplot[red, dashed, mark options=solid,mark = o]
table[x=snr,y=fer]{snr fer
3.5 0.116009280742459
4 0.0308356460067838
4.5 0.00511980339954946
5 0.000678251944372163
5.5 7.70829251685386e-05
};\addlegendentry{SCL, $L=64$}

\addplot[red, dashed, mark options=solid,mark = square]
table[x=snr,y=fer]{snr fer
3.5 0.0654450261780105
4 0.0180538003249684
4.5 0.00325796572620056
5 0.000328633869006540
5.5 3.18053707038284e-05
};\addlegendentry{SCL, $L=128$}

\addplot[brown, mark = +]
table[x=snr,y=fer]{snr fer
3.5 0.171232876712329
4 0.0645994832041344
4.5 0.0139662234998203
5 0.00260661036388281
5.5 0.000438406317300790
};\addlegendentry{SCL, $L=128,q=1$}

\addplot[brown, mark = x]
table[x=snr,y=fer]{snr fer
3.5 0.102459016393443
4 0.0301568154402895
4.5 0.00628851716765187
5 0.000994450963622984
5.5 0.000117971271635931
};\addlegendentry{SCL, $L=128,q=2$}

\addplot[brown, mark = o]
table[x=snr,y=fer]{snr fer
3.5 0.0805152979066023
4 0.0223713646532439
4.5 0.00405248447204969
5 0.000501811539658166
5.5 5.82543971819432e-05
};\addlegendentry{SCL, $L=128,q=3$}

\addplot[green, mark = +]
table[x=snr,y=fer]{snr fer
3.5 0.243902439024390
4 0.106382978723404
4.5 0.0297619047619048
5 0.00914578379367112
5.5 0.00179372197309417
};\addlegendentry{SCL, $L=32,q=1$}

\addplot[green, mark = x]
table[x=snr,y=fer]{snr fer
3.5 0.214592274678112
4 0.0665778961384820
4.5 0.0215796288303841
5 0.00415834996673320
5.5 0.000616834651303372
};\addlegendentry{SCL, $L=32,q=2$}

\addplot[green, mark = o]
table[x=snr,y=fer]{snr fer
3.5 0.151515151515152
4 0.0454959053685168
4.5 0.0116252034410602
5 0.00177733541874022
5.5 0.000192769597921173
};\addlegendentry{SCL, $L=32,q=3$}

\addplot[brown, mark = *, line width=0.5mm]
table[x=snr,y=fer]{snr fer
5.5 5.82543971819432e-05
};

\addplot[black, mark = *, line width=0.5mm]
table[x=snr,y=fer]{snr fer
5.5 7.8103e-05
};

\addplot[blue, mark = *, line width=0.5mm]
table[x=snr,y=fer]{snr fer
5.5 2.6209e-05
};

\end{semilogyaxis}
\end{tikzpicture}
	\caption{$(128,106+11)$ polar code with $11$ bits CRC ($\texttt{0x710}$) under \acp{DSGRAND} with $q\in\{1,2,3\}$, ORBGRAND and CA-SCL decoders with $L\in\{8,32,64,128\}$. The detailed computational complexity and memory requirement of three decoders at $E_b/N_0=5.5~\text{dB}$ (the highlighted points) are shown in Table~\ref{tab:complexity55}.}
	\label{fig:polarcrc_128_106}
\end{figure}
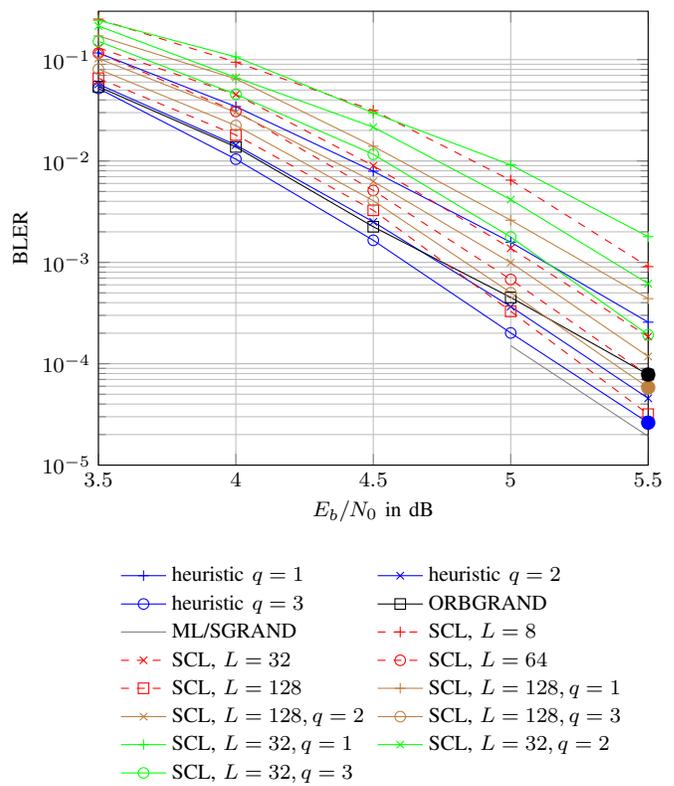

Figures~\ref{fig:bch_31_16}-\ref{fig:rs_155} present the decoding performance and average complexity (in number of guesses) of several \ac{BCH} and \ac{RS} codes of different length and code rate with \ac{DSGRAND}, basic \ac{ORBGRAND} and \ac{SGRAND}. We have following observations,
\begin{itemize}
    \item Basic \ac{ORBGRAND} performs close to \ac{SGRAND} at moderate \ac{SNR}, but is less precise at higher \ac{SNR}, which matches our analysis in Section~\ref{sec:orbgrand}.
    \item A \ac{GRAND} with \emph{better} performance requires \emph{fewer} guesses in average, since the computational power is used in a more efficient way.
    \item At moderate \ac{SNR}, mutual information (or $R_\text{LM}$) maximizing non-uniform quantizer provides the best performance. Interestingly, the heuristic by eq.~\eqref{eq:heuristic} outperforms the mutual information maximizing quantizer at high \ac{SNR}.
    \item The $3$ bits quantizer with heuristic approach provides near-\ac{ML} performance (within $0.1~\text{dB}$ loss) for all operating points.
\end{itemize}

Figure~\ref{fig:polarcrc_128_106} shows the performance of the \acp{GRAND} and \ac{CA-SCL} decoders. We consider a $(128,106+11)$ \ac{5G} polar code concatenated with $11$ bits \ac{CRC}. With a $3$ bits quantizer, proposed \ac{DSGRAND} performs $\approx 0.2~\text{dB}$ better than a \ac{CA-SCL} decoder with list size $L=128$.

\subsection{Computational complexity}
The main budget of computational power for a \ac{GRAND} is the parity check, which requires $n_p\cdot n\cdot (n-k)$ exclusive or operations in fully parallel implementations. Note that the complexity of parity check is much smaller in serial implementations, since the parity check function would return a $\texttt{false}$ once a unsatisfied equation is found, instead of check all $n-k$ equations. Additionally, both \ac{ORBGRAND} and \ac{DSGRAND} need some real additions and compare operations. \Ac{ORBGRAND} requires a size-$n$ sorting operation to get the rank information of the reliabilities. 

\subsection{Memory requirement}
For a \ac{DSGRAND}, the required memory space is the \ac{DP} table and the stack $\Xi$, i.e., $n\cdot\left(s_\text{max}+ |\Xi|\right)$ bits and $3\cdot|\Xi|$ rela numbers. In the our experiments, the size of $\Xi$ never exceeds $12$ for all codes, all quantization levels and all operating points. The selection of maximum score $s_\text{max}$ could be found in Section~\ref{sec:Smax} and is \ac{SNR}-related. We use $s_\text{max}=38$ for the \ac{DSGRAND} ($q=3$) at $E_b/N_0=5.5~\text{dB}$ shown in Figure~\ref{fig:polarcrc_128_106} and Table~\ref{tab:complexity55}. On the other side, an \ac{CA-SCL} requires $n\cdot L$ real numbers and $2n\cdot L$ bits memory space.

Table~\ref{tab:complexity55} shows the decoding complexity in (average) number of required operations and memory requirement of three decoders in Figure~\ref{fig:polarcrc_128_106} at $E_b/N_0=5.5~\text{dB}$. For \ac{DSGRAND} and \ac{ORBGRAND}, we use the fully parallel implementation, i.e., $n_p\cdot n\cdot (n-k)$ exclusive or operations are required. 

In Table~\ref{tab:complexity55}, we also provide the complexity and space score for the decoders, i.e., we assume that a real addition operation is eight times more complex than an exclusive or, a compare operation is six times more complex than an exclusive or, and the complexity of a size-$n$ sort (and find median) operation is close to $6\cdot n\cdot \log_2 n$ exclusive or operations. Under this assumption, the proposed \ac{DSGRAND} requires only a few percent of the complexity and space of a \ac{CA-SCL} decoder at $E_b/N_0=5.5~\text{dB}$.


\begin{table*}[t]
    \centering
    \setlength{\tabcolsep}{10pt}
    \renewcommand\arraystretch{1.5}
    \caption{Complexity in average number of operations and memory requirement of the decoders in Figure~\ref{fig:polarcrc_128_106}.}
\begin{tabular}{ |c||c|c||c|c| }
\hline
$E_b/N_0=5.5~\text{dB}$ & type of operation & \# of operations & type of data & required space \\
\hline
\hline
\multirow{6}{*}{\shortstack{CA-SCL, $L=128$\\ $\text{BLER}=5.825\times 10^{-5}$\\ \textcolor{red}{$\text{complexity score}=3226875$} \\ \textcolor{blue}{$\text{space score}=163840$}} } & exclusive or \textcolor{red}{$(1)$} & $55245$ &  &  \\ 
\cline{2-3}
& real addition \textcolor{red}{$(8)$} & $82143$& real number \textcolor{blue}{$(8)$} & $16384$\\ 
\cline{2-3}
& size-$2L$ find median \textcolor{red}{$(12288)$} & $117$&  &  \\
\cline{2-5}
& compare \textcolor{red}{$(6)$} & $69185$&   &  \\
\cline{2-3}
& real number copy \textcolor{red}{$(8)$} & $66168$& bit \textcolor{blue}{$(1)$} & $32768$\\
\cline{2-3}
& bit copy \textcolor{red}{$(1)$} & $132336$&  &  \\
\hline
\hline
\multirow{4}{*}{\shortstack{Basic ORBGRAND\\ $\text{BLER}=7.810\times 10^{-5}$\\ \textcolor{red}{$\text{complexity score}=489866$} \\ \textcolor{blue}{$\text{space score}=1152$}}} & exclusive or \textcolor{red}{$(1)$}& $449771$& real number \textcolor{blue}{$(8)$} & $128$ \\ 
\cline{2-3}
& size-$n$ sort \textcolor{red}{$(5397)$} & $1$& & \\ 
\cline{2-5}
& compare \textcolor{red}{$(6)$}& $3660$&  bit \textcolor{blue}{$(1)$} & $128$ \\ 
\cline{2-3}
& real addition \textcolor{red}{$(8)$} & $1592.2$&  &  \\ 
\hline
\hline
\multirow{5}{*}{\shortstack{DSGRAND, $q=3$, $s_\text{max}=38$\\ $\text{BLER}=2.621\times 10^{-5}$\\ \textcolor{red}{$\text{complexity score}=181316$} \\ \textcolor{blue}{$\text{space score}=7312$}}} & exclusive or \textcolor{red}{$(1)$} & $168011$&  & \\ 
\cline{2-3}
& compare \textcolor{red}{$(6)$} & $835$&  real number \textcolor{blue}{$(8)$} & $114$ \\ 
\cline{2-3}
& real addition \textcolor{red}{$(8)$} & $699.1$&  &  \\ 
\cline{2-5}
& real number copy \textcolor{red}{$(8)$} & $53.34$& bit \textcolor{blue}{$(1)$} & $6400$\\
\cline{2-3}
& bit copy \textcolor{red}{$(1)$} & $2275.8$&  &  \\
\hline
\end{tabular}
\vspace{10pt}
    \label{tab:complexity55}
\end{table*}

\section{Conclusions}\label{sec:conclusions}
Universal decoders have many practical benefits, including the ability to support
an arbitrary number of distinct codes with  one efficient piece of software or hardware, enabling the best choice of code for each application and future
proofing devices to the introduction of new codes. Particularly as new applications drive demand for shorter, higher rate error correcting codes, \Ac{GRAND} is a promising approach to realising this possibility. \Ac{GRAND} with soft input information handles codes for which there are only hard detection decoders, such as \ac{BCH} codes, or no established error correcting decoder, such as \acp{CRC}~\cite{an21, LiangLiu2021}, to soft detection decoding. Consistently with theoretical predictions~\cite{Coffey90}, results from \ac{GRAND} algorithms show that decoding performance is largely driven by the quality of the decoder rather than that of the code, and that good \ac{CRC} codes and codes selected at random can offer comparable performance
to highly structured ones~\cite{an20, an21, duffy2021ordered, Riaz21, Papadopoulou21}. 


In this work, we have analyzed the achievable rate of \ac{ORBGRAND}, \ac{ORBGRAND}, which requires a size-$n$ sorting operation on reliabilities, in effect at most  $\lceil\log_2(n)\rceil$ bits of  quantization. Our approach is based on order statistics and mismatched decoding. We have shown that multi-line \ac{ORBGRAND} can be capacity achieving at all \ac{SNR} points with precise enough curve fitting through multi-line models. Its complexity is also far lower than \ac{CA-SCL}, both in terms of and computation and, in particular, memory. 

We have  introduced \ac{DSGRAND}, that uses a conventional quantizer. \Ac{DSGRAND} inherits all the desirable features of \ac{GRAND} algorithms, including universality, parallelizability and reduced algorithmic effort as \ac{SNR} increases.  It can avail of any level of quantization, provides improved error correction performance as quantization level increases, and obviates the need for a sorter. The numerical result shows that \ac{DSGRAND} outperforms basic \ac{ORBGRAND} at high \ac{SNR} and performs within $0.25~\text{dB}$ and $0.1~\text{dB}$ close to \ac{ML} decoding with $2$ and $3$ bits quantizers, respectively. \ac{DSGRAND} outperforms \ac{CA-SCL} with sharply lower complexity and memory requirements. With respect to basic \ac{ORBGRAND}, it has somewhat lower complexity but somewhat higher memory requirements.


\bibliographystyle{IEEEtran}
\bibliography{ORBGRAND,grand}

\begin{thebibliography}{10}
\providecommand{\url}[1]{#1}
\csname url@samestyle\endcsname
\providecommand{\newblock}{\relax}
\providecommand{\bibinfo}[2]{#2}
\providecommand{\BIBentrySTDinterwordspacing}{\spaceskip=0pt\relax}
\providecommand{\BIBentryALTinterwordstretchfactor}{4}
\providecommand{\BIBentryALTinterwordspacing}{\spaceskip=\fontdimen2\font plus
\BIBentryALTinterwordstretchfactor\fontdimen3\font minus
  \fontdimen4\font\relax}
\providecommand{\BIBforeignlanguage}[2]{{%
\expandafter\ifx\csname l@#1\endcsname\relax
\typeout{** WARNING: IEEEtran.bst: No hyphenation pattern has been}%
\typeout{** loaded for the language `#1'. Using the pattern for}%
\typeout{** the default language instead.}%
\else
\language=\csname l@#1\endcsname
\fi
#2}}
\providecommand{\BIBdecl}{\relax}
\BIBdecl

\bibitem{berlekamp1978inherent}
E.~Berlekamp, R.~McEliece, and H.~Van~Tilborg, ``On the inherent intractability
  of certain coding problems (corresp.),'' \emph{IEEE Tran. Inf. Theory},
  vol.~24, no.~3, pp. 384--386, 1978.

\bibitem{berlekamp1968algebraic}
E.~Berlekamp, \emph{Algebraic coding theory}.\hskip 1em plus 0.5em minus
  0.4em\relax World Scientific, 1968.

\bibitem{massey1969shift}
J.~Massey, ``Shift-register synthesis and {BCH} decoding,'' \emph{IEEE Trans.
  Inf Theory}, vol.~15, no.~1, pp. 122--127, 1969.

\bibitem{gallager1963low}
R.~G. Gallager, ``Low density parity check codes,'' 1963.

\bibitem{chen2002near}
J.~Chen and M.~P. Fossorier, ``Near optimum universal belief propagation based
  decoding of low-density parity check codes,'' \emph{IEEE Trans. Commun},
  vol.~50, no.~3, pp. 406--414, 2002.

\bibitem{niu2012crc}
K.~Niu and K.~Chen, ``{CRC}-aided decoding of {P}olar codes,'' \emph{IEEE
  Commun. Letters}, vol.~16, no.~10, pp. 1668--1671, 2012.

\bibitem{tal2015list}
I.~Tal and A.~Vardy, ``List decoding of {P}olar codes,'' \emph{IEEE Trans. Inf.
  Theory}, vol.~61, no.~5, pp. 2213--2226, 2015.

\bibitem{balatsoukas2015llr}
A.~Balatsoukas-Stimming, M.~B. Parizi, and A.~Burg, ``{LLR}-based successive
  cancellation list decoding of {P}olar codes,'' \emph{IEEE Trans. Signal
  Process.}, vol.~63, no.~19, pp. 5165--5179, 2015.

\bibitem{leonardon2019fast}
M.~Leonardon, A.~Cassagne, C.~Leroux, C.~Jego, L.-P. Hamelin, and Y.~Savaria,
  ``Fast and flexible software polar list decoders,'' \emph{J. Signal Process.
  Syst.}, pp. 1--16, 2019.

\bibitem{durisi2016toward}
G.~Durisi, T.~Koch, and P.~Popovski, ``Toward massive, ultrareliable, and
  low-latency wireless communication with short packets,'' \emph{Proc. IEEE},
  vol. 104, no.~9, pp. 1711--1726, 2016.

\bibitem{she2017radio}
C.~She, C.~Yang, and T.~Q. Quek, ``Radio resource management for ultra-reliable
  and low-latency communications,'' \emph{IEEE Commun. Mag.}, vol.~55, no.~6,
  pp. 72--78, 2017.

\bibitem{chen2018ultra}
H.~Chen, R.~Abbas, P.~Cheng, M.~Shirvanimoghaddam, W.~Hardjawana, W.~Bao,
  Y.~Li, and B.~Vucetic, ``Ultra-reliable low latency cellular networks: Use
  cases, challenges and approaches,'' \emph{IEEE Commun. Mag.}, vol.~56,
  no.~12, 2018.

\bibitem{parvez2018survey}
I.~Parvez, A.~Rahmati, I.~Guvenc, A.~I. Sarwat, and H.~Dai, ``A survey on low
  latency towards {5G: RAN}, core network and caching solutions,'' \emph{IEEE
  Commun. Surv.}, vol.~20, no.~4, pp. 3098--3130, 2018.

\bibitem{medard20205}
M.~M{\'e}dard, ``Is 5 just what comes after 4?'' \emph{Nature Electronics},
  vol.~3, no.~1, pp. 2--4, 2020.

\bibitem{Duffy18}
K.~R. Duffy, J.~Li, and M.~M\'edard, ``Guessing noise, not code-words,'' in
  \emph{IEEE Int. Symp. on Inf. Theory}, 2018.

\bibitem{duffy19GRAND}
K.~R. {Duffy}, J.~{Li}, and M.~{M\'edard}, ``Capacity-achieving guessing random
  additive noise decoding,'' \emph{IEEE Trans. Inf. Theory}, vol.~65, no.~7,
  pp. 4023--4040, 2019.

\bibitem{abbas2020}
S.~M. Abbas, T.~Tonnellier, F.~Ercan, and W.~J. Gross, ``High-throughput {VLSI}
  architecture for {GRAND},'' in \emph{IEEE SiPS}, 2020.

\bibitem{Riaz21}
A.~Riaz, V.~Bansal, A.~Solomon, W.~An, Q.~Liu, K.~Galligan, K.~R. Duffy,
  M.~M\'edard, and R.~T. Yazicigil, ``Multi-code multi-rate universal maximum
  likelihood decoder using {GRAND},'' in \emph{IEEE ESSCIRC}, 2021.

\bibitem{Riaz22}
A.~Riaz, M.~Medard, K.~R. Duffy, and R.~T. Yazicigil, ``A universal maximum
  likelihood grand decoder in 40nm {CMOS},'' in \emph{COMSNETS}, 2022, pp.
  421--423.

\bibitem{an20}
W.~An, M.~M{\'e}dard, and K.~R. Duffy, ``Keep the bursts and ditch the
  interleavers,'' in \emph{IEEE GLOBECOM}, 2020.

\bibitem{an2022keep}
------, ``Keep the bursts and ditch the interleavers,'' \emph{IEEE Trans.
  Commun.}, vol.~70, no.~6, pp. 3655--3667, 2022.

\bibitem{abbas2021high-MO}
S.~M. Abbas, M.~Jalaleddine, and W.~J. Gross, ``High-throughput {VLSI}
  architecture for {GRAND Markov Order},'' in \emph{IEEE SiPS}, 2021, pp.
  158--163.

\bibitem{zhan2021noise}
M.~Zhan, Z.~Pang, K.~Yu, J.~Xu, F.~Wu, and M.~Xiao, ``Noise error pattern
  generation based on successive addition-subtraction for {GRAND-MO},''
  \emph{IEEE Commun. Letters}, vol.~26, no.~4, pp. 743--747, 2022.

\bibitem{solomon20SGRAND}
A.~Solomon, K.~R. Duffy, and M.~M\'edard, ``Soft maximum likelihood decoding
  using {GRAND},'' in \emph{IEEE ICC}, 2020.

\bibitem{Duffy19a}
K.~R. Duffy and M.~M\'edard, ``Guessing random additive noise decoding with
  soft detection symbol reliability information,'' in \emph{IEEE Int. Symp. on
  Inf. Theory}, 2019.

\bibitem{Duffy22}
K.~R. Duffy, M.~Médard, and W.~An, ``Guessing random additive noise decoding
  with symbol reliability information ({SRGRAND}),'' \emph{IEEE Trans.
  Commun.}, vol.~70, no.~1, pp. 3--18, 2022.

\bibitem{duffy2021ordered}
K.~R. Duffy, ``Ordered reliability bits guessing random additive noise
  decoding,'' in \emph{IEEE ICASSP}, 2021, pp. 8268--8272.

\bibitem{abbas2021high}
S.~M. Abbas, T.~Tonnellier, F.~Ercan, M.~Jalaleddine, and W.~J. Gross,
  ``High-throughput {VLSI} architecture for soft-decision decoding with
  {ORBGRAND},'' in \emph{IEEE ICASSP}, 2021, pp. 8288--8292.

\bibitem{condo2021high}
C.~Condo, V.~Bioglio, and I.~Land, ``High-performance low-complexity error
  pattern generation for {ORBGRAND} decoding,'' in \emph{IEEE GLOBECOM}, 2021.

\bibitem{condo2022fixed}
C.~Condo, ``A fixed latency {ORBGRAND} decoder architecture with {LUT}-aided
  error-pattern scheduling,'' \emph{IEEE Trans. Circuits Syst. I Regul. Pap.},
  2022.

\bibitem{abbas22}
S.~M. Abbas, T.~Tonnellier, F.~Ercan, M.~Jalaleddine, and W.~J. Gross,
  ``High-throughput and energy-efficient vlsi architecture for ordered
  reliability bits grand,'' \emph{IEEE Trans. Very Large Scale Integr. (VLSI)
  Syst.}, pp. 1--13, 2022.

\bibitem{liu2022orbgrand}
M.~Liu, Y.~Wei, Z.~Chen, and W.~Zhang, ``{ORBGRAND} is almost
  capacity-achieving,'' \emph{arxiv:2202.06247}, 2022.

\bibitem{abbas2021list}
S.~M. Abbas, M.~Jalaleddine, and W.~J. Gross, ``{List-GRAND}: A practical way
  to achieve maximum likelihood decoding,'' \emph{arXiv:2109.12225}, 2021.

\bibitem{duffy22ORBGRAND}
K.~R. Duffy, W.~An, and M.~M{\'e}dard, ``Ordered reliability bits guessing
  random additive noise decoding,'' \emph{IEEE Trans. Signal Process.},
  vol.~70, pp. 4528--4542, 2022.

\bibitem{boucheron2012concentration}
S.~Boucheron and M.~Thomas, ``Concentration inequalities for order
  statistics,'' \emph{Electron. Commun. Probab.}, vol.~17, pp. 1--12, 2012.

\bibitem{bocherer2018principles}
G.~B{\"o}cherer, ``Principles of coded modulation,'' Habilitation, TU
  M{\"u}nchen, 2018.

\bibitem{ganti2000mismatched}
A.~Ganti, A.~Lapidoth, and I.~E. Telatar, ``Mismatched decoding revisited:
  General alphabets, channels with memory, and the wide-band limit,''
  \emph{IEEE Trans. Inform. Theory}, vol.~46, no.~7, pp. 2315--2328, 2000.

\bibitem{bocherer2017efficient}
G.~B\"ocherer, T.~Prinz, P.~Yuan, and F.~Steiner, ``Efficient polar code
  construction for higher-order modulation,'' in \emph{IEEE WCNCW}, 2017.

\bibitem{kaplan1993information}
G.~Kaplan and S.~Shamai, ``Information rates and error exponents of compound
  channels with application to antipodal signaling in a fading environment,''
  \emph{Arch. Elektron. Uebertrag.}, vol.~47, no.~4, pp. 228--239, 1993.

\bibitem{kleinberg2006algorithm}
J.~Kleinberg and E.~Tardos, \emph{Algorithm design}.\hskip 1em plus 0.5em minus
  0.4em\relax Pearson Education India, 2006.

\bibitem{forney1968exponential}
G.~Forney, ``Exponential error bounds for erasure, list, and decision feedback
  schemes,'' \emph{IEEE Trans. Inf. Theory}, vol.~14, no.~2, pp. 206--220,
  1968.

\bibitem{hof2010performance}
E.~Hof, I.~Sason, and S.~Shamai, ``Performance bounds for erasure, list, and
  decision feedback schemes with linear block codes,'' \emph{IEEE Trans. Inf.
  Theory}, vol.~56, no.~8, pp. 3754--3778, 2010.

\bibitem{an21}
W.~{An}, K.~R. {Duffy}, and M.~{M\'edard}, ``{CRC} codes as error correction
  codes,'' in \emph{IEEE ICC}, 2021.

\bibitem{LiangLiu2021}
W.~Liang and H.~Liu, ``Low-complexity error correction algorithm for cyclic
  redundancy check codes,'' in \emph{IEEE ICCC}, 2021, pp. 22--26.

\bibitem{Coffey90}
J.~T. Coffey and R.~M. Goodman, ``Any code of which we cannot think is good,''
  \emph{IEEE Trans. Inf. Theory}, vol.~36, no.~6, pp. 1453--1461, 1990.

\bibitem{Papadopoulou21}
V.~Papadopoulou, M.~Hashemipour-Nazari, and A.~Balatsoukas-Stimming, ``Short
  codes with near-{ML} universal decoding: are random codes good enough?'' in
  \emph{IEEE SiPS}, 2021, pp. 94--98.

\end{thebibliography}

\end{document}